%% file: main.tex
\documentclass[10pt]{article}
\usepackage{fullpage}
\usepackage{amsmath,amsfonts,amsthm,amssymb,dsfont,bbm,bm}
\usepackage{url}
\usepackage{graphicx}
\usepackage{wrapfig}
\usepackage{caption} 
\usepackage{subcaption}
\usepackage{algorithm}
\usepackage[noend]{algpseudocode}
\usepackage{float}
\usepackage{framed}
\usepackage{comment}
\usepackage{enumerate}
\usepackage{xcolor}\usepackage[colorlinks=true, linkcolor=red, urlcolor=blue, citecolor=blue]{hyperref}
\usepackage[makeroom]{cancel}
\usepackage[shortlabels,inline]{enumitem}
\usepackage{multirow}
\usepackage{booktabs} 
\usepackage{array}
\usepackage{braket}
\usepackage{authblk}  
\usepackage{mathrsfs}
\usepackage{circuitikz}
\usepackage{tikz}
\newcounter{tableeqn}[table]

\DeclareMathOperator{\Tr}{Tr}
\DeclareMathOperator{\supp}{supp}
\DeclareMathOperator*{\argmin}{argmin}
\DeclareMathOperator*{\diag}{diag}
\newtheorem{claim}{Claim}
\newtheorem{corollary}{Corollary}
\newtheorem{fact}{Fact}
\newtheorem{lemma}{Lemma}

\theoremstyle{definition}
\newtheorem{definition}{Definition}
\newtheorem{problem}{Problem}

\newenvironment{proofsketch}{%
\proof}{\endproof}

\newcommand{\sizeof}[1]{\left\lvert{#1}\right\rvert}

\newcommand*\circled[1]{\tikz[baseline=(char.base)]{\node[shape=circle,fill,inner sep=2pt] (char) {\textcolor{white}{#1}};}}

\title{Quantum-enhanced Network Tomography}

\author[1]{Yufei Zheng\thanks{\texttt{\{yufeizheng, towsley\}@cs.umass.edu}}}
\author[2]{Zihao Gong\thanks{\texttt{\{zgong12, saikat\}@umd.edu}}}
\author[2]{Saikat Guha}
\author[1]{Don Towsley}

\affil[1]{\small\textit{College of Information and Computer Sciences, University of Massachusetts Amherst}}
\affil[2]{\small\textit{Department of Electrical and Computer Engineering, University of Maryland, College Park}}

\date{}  

\begin{document}
\maketitle
\begin{abstract}
\emph{Network tomography} refers to the use of inference techniques for inferring internal network states from end-to-end probes.
\emph{Quantum probes}, implemented by sending blocks of $n$ coherent-state pulses augmented with continuous-variable (CV) squeezing ($n=1$) or weak temporal-mode entanglement ($n>1$) over a lossy channel to a receiver with homodyne detection capabilities, are known to carry information about the channel transmissivity.
Assuming a subset of nodes in an optical network is capable of sending and receiving such probes through intermediate nodes with all-optical switching capabilities, we leverage these quantum probes to estimate link transmissivities.

To determine how to route the probes in a network, we propose a probe construction algorithm that guarantees link \emph{identifiability}, while maximizing the number of \emph{information orthogonal} sets of transmissivities.
A set of probes induces a \emph{Fisher information matrix (FIM)}.
We then derive two metrics, the determinant of the FIM and the trace of its inverse, to evaluate the performance of the probes.
In particular, our results can be used to characterize the quantum improvement in estimating link transmissivities in a general optical network.
\end{abstract}

\input{intro}
\input{prelim}
\input{physical}
\input{algo}
\input{general}
\input{related}
\input{conclusion}

\bibliographystyle{plain}
\newpage
\bibliography{ref}

\newpage
\appendix
\input{appendix}

\end{document}

%% file: intro.tex
\section{Introduction}

Network tomography~\cite{vardi1996network,he2021network} refers to the use of inference techniques that reconstruct the \emph{internal} states of a network using \emph{external} measurements, taken between a selected subset of network nodes.
In stark contrast to conventional approaches that require link- or device-level access to every internal state of interest (the Simple Network Management Protocol (SNMP)~\cite{case1989simple}), network tomography obviates this need by solely leveraging end-to-end measurements, making it particularly suitable for monitoring networks with inaccessible internal nodes (optical networks~\cite{harvey2007non}, cloud networks~\cite{chard2016network}). 
When measurements are modeled as random variables, network tomography bears some resemblance to multiparameter estimation. 
While certain statistical techniques transfer, the key difference lies in that network tomography problems are fundamentally graph-constrained: the underlying network topology must be accounted for when designing measurements.

A large body of existing work in classical network tomography has focused on upper-layer performance metrics such as packet loss rate~\cite{he2015fisher,ghita2010netscope,bu2002network,tsang2001passive}, delay~\cite{he2015fisher,duffield2004network,tsang2003network,coates2001network}, and bandwidth~\cite{dichev2012efficient,bobelin2008algorithms}.
Quantum network tomography, a newly emerged field of study, applies the classical tomography framework to quantum networks.
By treating a quantum network as a collection of noisy quantum channels, recent work has investigated how to determine the corresponding channel parameters~\cite{wang2025quantum,de2024quantum}.

In this work, we consider the problem of estimating link transmissivity, a physical-layer metric in optical networks.
This problem may arise in several contexts, as optical networks are not only highly prevalent in modern communication systems~\cite{korsback2023growing,velush2023boosting,poutievski2022jupiter}, but also strong candidates for the physical infrastructure of future quantum networks~\cite{nellis2026quantum,lukens2025hybrid}.
Meanwhile, link transmissivity is central to optical power budgeting~\cite{liang2022link} and routing decisions~\cite{gao2020resource}, while also serving as a critical indicator for performance monitoring and fault detection in optical networks~\cite{zheng2025quantum}.
To estimate the transmissivities of all links, network tomography is a natural approach. 
Even with full control over all the links, exhaustively testing each one does not scale to large networks. 
Instead, we rely on end-to-end \emph{probes} to infer the transmissivities.

\emph{Quantum probes}, implemented by sending blocks of $n$ coherent-state pulses augmented with continuous-variable (CV) squeezing or weak temporal-mode entanglement over a pure-loss bosonic channel to a receiver with homodyne detection capabilities, are known to be more sensitive than their quasi-classical counterparts in detecting sudden changes in transmission loss~\cite{guha2025quantum,zheng2025quantum}.
These probes, when traversing a set of links in the network, carry information about link transmissivities. 
Therefore, we also adopt such probes in solving our estimation problem.

Although transceivers for quantum probes are experimentally accessible using current quantum optical technology, it is only reasonable to assume that a \emph{subset} of network nodes will be equipped with such capabilities.
Given a network topology and a subset of network nodes that can send and receive quantum probes, to infer transmissivities for all links, we must decide how to route the probes in the network, and what physical implementation to use for each probe.
Solving the latter involves optimizing parameters for the physical implementation and experiment design. 
While it is an interesting problem in its own right, it lies beyond the scope of the present work.
Instead, we focus on routing, and on characterizing how well a given set of probes with specified physical implementations performs the estimation task.

In deciding how to route the probes in a network, we incorporate the notion of \emph{information orthogonality} from estimation theory~\cite{cox1987parameter}.
Informally, and particularly in the context of network tomography, two sets of unknown parameters that are information orthogonal can be estimated separately.
Estimating unknown link transmissivities in practice for a general network typically involves the use of numerical solvers.
Subsets of link transmissivities that are information orthogonal can be treated independently and in parallel using numerical tools, which makes solving large instances more tractable.
Therefore, we introduce a probe construction algorithm that maximizes the number of information orthogonal subsets of link transmissivities.

To quantify how much information a given set of probes $\mathscr{P}$ carries about the unknown link transmissivities, we consider the \emph{Fisher information matrix (FIM)} induced by $\mathscr{P}$.
Since the observation from each probe can be modeled as a Gaussian random vector~\cite{guha2025quantum}, the FIM admits a convenient structure that enables closed-form expressions for both its determinant and the trace of its inverse, two standard metrics from the theory of optimal experiment design~\cite{atkinson2007optimum} for evaluating the performance of $\mathscr{P}$.
Moreover, we show that, in a general network setting, both metrics are closely related to the Fisher information (FI) associated each individual probe, which allows us to extend results from the single-channel case to characterize the quantum improvement for the general network case.

Our contributions in this paper are twofold:
\begin{itemize}
\item We present a probe construction algorithm that guarantees the identifiability of all link transmissivities, while maximizing the number of information orthogonal subsets of these unknown parameters (\S~\ref{sec:graph}).
\item We derive the determinant of the FIM induced by a given set of probes, and the trace of its inverse, and use both metrics to characterize the quantum improvement in a general network (\S~\ref{sec:combined}).
\end{itemize}
In \S~\ref{sec:prelim}, we formally define what constitutes a probe, and provide the necessary background on the two metrics related to the FIM.
\S~\ref{sec:physical} reviews the quantum and classical probes used throughout this work, and discusses the potential benefits---or lack thereof---of sharing entanglement across multiple channels in the network.
Finally, we discuss related work in \S~\ref{sec:related}, and conclude the paper in \S~\ref{sec:conclusion}.

%% file: prelim.tex
\section{Preliminaries} \label{sec:prelim}

A network is a graph $G = (V(G),E(G))$, with a set of \emph{monitors} $M \subseteq V(G)$ that can \emph{send} and \emph{receive} probes.
Here each edge (link) $e_i$ in $E(G)$ represents a channel with transmissivity $\eta_i$.
Throughout this work, we assume the graph topology is \emph{known} (i.e., $G$ and $M$), but the transmissivity of each link is \emph{unknown}, and the problem is to estimate the vector of all link transmissivities $\bm{\eta}=(\eta_1, \eta_2, \dots, \eta_n)$, through probing the network.
Next we formally define what a probe is (\S~\ref{sec:prelim:probe}), and introduce the two metrics  for characterizing the performance of the probes (\S~\ref{sec:prelim:metrics}).

\subsection{What is a probe?} \label{sec:prelim:probe}

\paragraph{Definition.}
Formally, a probe is a $4$-tuple $(P, \texttt{impl}(P), t(P), c(P))$, where $P$ describes how the probe is routed in the network, $\texttt{impl}(P) \in \{ `\verb|coherent|\textrm', `\verb|squeezed|\textrm', `\verb|entangled|\textrm'\}$ specifies the physical implementation\footnote{This is an oversimplification for notational convenience. $\texttt{impl}(P)$ should also include the classical coherent energy $N$ and the quantum energy $N_a$ used in the physical implementation. In what follows, assume these parameters are explicitly specified in $\texttt{impl}(P)$.}, $t(P)$ is the number of pulses in the entangled block, and $c(P)$ is the number of copies to send.
For example, we may send $3$ copies of a probe that follows the route $P$, and each copy consists of a block of $4$ entangled pulses.
In this example, $\texttt{impl}(P) = `\verb|entangled|\textrm'$, $t(P) = 4$, and $c(P) = 3$.
If a probe that traverses some $P'$ is implemented by coherent of squeezed states, $t(P')$ is trivially set to $1$.
Here $P$, the graph-theoretic description of a probe in the $4$-tuple representation, is formally a \emph{walk}\footnote{A \emph{walk} with endpoints $v_0$ and $v_k$ is a sequence $v_0, e_1, v_1, e_2, ..., e_k, v_k$, where $v_i \in V$ for all $0\leq i\leq k$, and $e_j \in E$ for all $1\leq i\leq k$.} in the graph with endpoints in the set of monitors $M$.
However, it is more convenient to think of $P$ as the \emph{multiset} of edges that the probe traverses.
Whenever we refer to a set of probes $\mathscr{P}$, we mean that $\mathscr{P} = \{(P_i, \texttt{impl}(P_i), t(P_i), c(P_i))\}_{i=1}^n$, where $P_1, P_2, \dots, P_n$ are distinct.

For ease of notation, the term ``probe'' refers to different components in the formal $4$-tuple definition depending on the context.
In \S~\ref{sec:physical}, where the underlying graph is trivially an edge or a star, we treat a probe as $(\texttt{impl}(P), t(P), c(P))$.
In \S~\ref{sec:graph}, which is concerned solely with routing, a probe refers to the graph-theoretic object $P$.
Only in \S~\ref{sec:combined}, when all elements from previous sections are combined, do we require the full $4$-tuple definition of a probe.

\paragraph{Observations from probes.}
The three types of probes we consider in this work are implemented by coherent states, squeezed states, and entanglement augmented states. 
At the homodyne detector placed at the receiving end of these probes, each probe $(P, \texttt{impl}(P), t(P), c(P))$ produces a random vector $\bm{X}_P$ of length $t(P)c(P)$ sampled from some Gaussian distribution $\mathcal{N}(\bm{\mu}(\eta_P), \bm{\Sigma}(\eta_P))$, where\footnote{In this notation, the multiset $P$ is represented as sets, with some elements repeated. And we iterate over all edges $e$ in $P$, including repetitions.} $\eta_P = \prod_{e \in P} \eta_e$.
We refer to $\bm{X}_P$ as the \emph{observation} from the probe $(P, \texttt{impl}(P), t(P), c(P))$.
The observation from a set of probes $\mathscr{P} = \{(P_i, \texttt{impl}(P_i), t(P_i), c(P_i))\}_{i=1}^n$ is a random vector $\bm{X}$, given by the concatenation of $\bm{X}_{P_1}, \bm{X}_{P_2}, \dots, \bm{X}_{P_n}$, and sampled from some $d$-variate Gaussian distribution $\mathcal{N}_d(\bm{\mu}(\bm{\eta}), \bm{\Sigma}(\bm{\eta}))$, with mean vector $\bm{\mu}(\bm{\eta})$, covariance matrix $\bm{\Sigma}(\bm{\eta})$, and dimension $d = \sum_{i=1}^n t(P_i) c(P_i)$.
We detail the forms of $\bm{\mu}(\cdot)$ and $\bm{\Sigma}(\cdot)$ for different physical implementations in \S~\ref{sec:physical}.

\subsection{Two metrics related to the Fisher information matrix} \label{sec:prelim:metrics}

\paragraph{Fisher information matrix.}
In estimation theory, \emph{Fisher information matrix (FIM)} quantifies how much information the observations from the probes contain about the vector of unknown parameters $\bm{\eta}$.
Given probes $\mathscr{P}$, the vector of observations generated by $\mathscr{P}$ is sampled from a multivariate Gaussian distribution.
Therefore, we restrict our attention to the FIM of Gaussian distributions.
For a $d$-dimensional Gaussian distribution $\mathcal{N}_d(\bm{\mu}(\bm{\eta}), \, \bm{\Sigma}(\bm{\eta}))$, the $(i,j)$-entry of the FIM $\mathcal{I}(\bm{\eta})$ is 
\begin{align} \label{eqn:FIM}
\mathcal{I}_{i,j}
= \frac{\partial \bm{\mu}^T}{\partial \eta_i} \bm{\Sigma}^{-1} \frac{\partial \bm{\mu}}{\partial \eta_j}  + \frac{1}{2}\Tr \left(\bm{\Sigma}^{-1} \frac{\partial \bm{\Sigma}}{\partial \eta_i} \bm{\Sigma}^{-1}\frac{\partial \bm{\Sigma}}{\partial \eta_j}\right),
\end{align}
where
\begin{align*}
\frac{\partial \bm{\mu}}{\partial \eta_i}
=
\begin{pmatrix}
\frac{\partial \mu_1}{\partial \eta_i} \\
\frac{\partial \mu_2}{\partial \eta_i} \\
\vdots \\
\frac{\partial \mu_n}{\partial \eta_i}
\end{pmatrix},
\frac{\partial \bm{\Sigma}}{\partial \eta_i}
=
\begin{pmatrix}
\frac{\partial \Sigma_{1,1}}{\partial \eta_i} & \frac{\partial \Sigma_{1,2}}{\partial \eta_i} & \cdots & \frac{\partial \Sigma_{1,n}}{\partial \eta_i} \\
\frac{\partial \Sigma_{2,1}}{\partial \eta_i} & \frac{\partial \Sigma_{2,2}}{\partial \eta_i} & \cdots & \frac{\partial \Sigma_{2,n}}{\partial \eta_i} \\
\vdots & \vdots & \ddots & \vdots \\
\frac{\partial \Sigma_{n,1}}{\partial \eta_i} & \frac{\partial \Sigma_{n,2}}{\partial \eta_i} & \cdots & \frac{\partial \Sigma_{n,n}}{\partial \eta_i} \\
\end{pmatrix}.
\end{align*}
Note that the size of the FIM $\mathcal{I}(\bm{\eta})$ depends on the number of unknown parameters $n$, \emph{not} the dimensionality $d$ of the Gaussian distribution.
In particular, when the number of edges (i.e., the number of unknown link transmissivities) $n=1$, we have a single channel with unknown transmissivity $\eta$, and \eqref{eqn:FIM} gives the Fisher information (FI) $I(\eta)$.

\paragraph{Metrics.}
For an unbiased estimator $\hat{\bm{\eta}}$ of $\bm{\eta}$ that achieves the Cramér-Rao bound, the volume of its error ellipsoid is proportional to $\frac{1}{\sqrt{\det(\mathcal{I}(\bm{\eta}))}}$, and the sum of the variances of the individual components of $\hat{\bm{\eta}}$ is equal to $\Tr(\mathcal{I}(\bm{\eta})^{-1})$.
Even if an efficient unbiased estimator does not exist, $\det(\mathcal{I}(\bm{\eta}))$ still characterizes the minimal achievable volume of the uncertainty ellipsoid, and $\Tr(\mathcal{I}(\bm{\eta})^{-1})$ represents a fundamental local information-theoretic benchmark for the minimum attainable total variance.
Therefore, we adopt $\det(\mathcal{I}(\bm{\eta}))$ and $\Tr(\mathcal{I}(\bm{\eta})^{-1})$ as measures for the performance a given set of probes $\mathscr{P}$.
Ideally, a good set of probes should induce a \emph{large} determinant and a \emph{small} trace.
In the special case of a single channel (\S~\ref{sec:physical:single}), $\det(\mathcal{I}(\bm{\eta})) = I(\eta)$, $\Tr(\mathcal{I}(\bm{\eta})^{-1}) = I(\eta)^{-1}$, and it suffices to consider the FI $I(\eta)$.

%% file: physical.tex
\section{Physical implementations of probes}
\label{sec:physical}

We begin this section by introducing the classical (coherent states) and quantum probes (squeezing and entanglement augmented states) considered throughout this work.
In \S~\ref{sec:physical:single}, we show that quantum probes are strictly better than classical ones for estimating the transmissivity of a single channel, and entanglement always brings benefits over squeezing alone.
With the goal of extending the channel case to a general network setting, we consider in \S~\ref{sec:physical:multi} the problem of estimating the transmissivities of multiple channels.
We demonstrate that, in contrast, sharing entanglement leaves both the determinant of the FIM and the trace of its inverse worse off than using squeezing alone.

\subsection{The single-channel case} \label{sec:physical:single}

Consider a single lossy channel between monitors $M_1$ and $M_2$, with unknown transmissivity $\eta$.
Fix the total number of probes $\sizeof{\mathcal{P}} = n$, where each probe is sent from $M_1$ to $M_2$.
At the receiving monitor, we observe an $n$-dimensional Gaussian random vector $\bm{X}\sim \mathcal{N}_n(\bm{\mu}(\eta), \bm{\Sigma}(\eta))$.
In what follows, we introduce specific forms of $\bm{\mu}(\cdot)$ and $\bm{\Sigma}(\cdot)$ for each physical implementation of $\mathcal{P}$, and compare the corresponding Fisher information $\mathcal{I}(\eta)$ between different implementations.

\paragraph{Coherent states.}
For the classical (or quasi-classical) benchmark, similar to~\cite{guha2025quantum,zheng2025quantum},
each probe is implemented as a coherent-state pulse $\ket{\alpha}$, with mean photon number $\alpha^2 = N + N_a$.
One can think of $N$ as the classical energy, while the quantum energy $0 < N_a << N$ is introduced to fairly compare the classical case with its quantum-augmented counterpart.
We assume all $n$ probes are independent.
Denote $\bm{I}$ as the $n \times n$ identity matrix, and $\bm{u}$ as the column vector of $n$ ones.
When we send $n$ coherent state pulses, the observation follows a Gaussian distribution with mean vector $\bm{\mu}_{\mathrm{c}}(\eta)$ and covariance matrix $\bm{\Sigma}_{\mathrm{c}}(\eta)$ specified below:
\begin{align}
& \bm{\mu}_{\mathrm{c}}(\eta) = \sqrt{(N+N_a)\eta} \cdot \bm{u}, \nonumber \\
& \bm{\Sigma}_{\mathrm{c}}(\eta) = \frac{1}{4} \bm{I}. \nonumber
\end{align}
Since $\bm{\Sigma}_{\mathrm{c}}(\eta)$ is constant with respect to $\eta$, the second term in~\eqref{eqn:FIM} is $0$, the Fisher information (FI) for the classical case, $I_{\mathrm{c}}(\eta)$, has the form
\begin{align}
& I_{\mathrm{c}}(\eta) = \frac{(N+N_a)n}{\eta}. \nonumber
\end{align}

\paragraph{Squeezed states.}
For the first quantum case, each probe is physically a displaced squeezed state $\ket{\alpha;r}$, with $\alpha^2 = N$ photon equivalent of classical coherent energy and $\sinh^2(r) = N_a << N$ quantum photon energy.
In both the classical and the squeezing cases, the $n$ pulses are uncorrelated, so again the covariance matrix for the observation is diagonal:
\begin{align}
& \bm{\mu}_{\mathrm{s}}(\eta) = \sqrt{N\eta} \cdot \bm{u}, \label{eqn:mu_sq} \\
& \bm{\Sigma}_{\mathrm{s}}(\eta) = \frac{1}{4}(1-(1-e^{-2r})\eta) \cdot \bm{I}. \label{eqn:Sigma_sq}
\end{align}
By~(\ref{eqn:FIM}), the FI for the squeezing case $I_{\mathrm{s}}(\eta)$ is 
\begin{equation} \label{eqn:I_s_with_r}
I_{\mathrm{s}}(\eta) = n \cdot \left( \frac{N}{\eta(1-(1-e^{-2r})\eta)} + \frac{(1-e^{-2r})^2}{2(1-(1-e^{-2r})\eta)^2} \right).
\end{equation}

\paragraph{Entanglement augmentation.}
Next, we consider entanglement-augmented probes.
Specifically, a block of $n \in \mathbb{Z}_+$ coherent-state pulses $\ket{\alpha}^{\otimes n}$ are augmented by a continuous variable entangled state generated by splitting a squeezed-vacuum state $\ket{0;s}$, $s \in \mathbb{R}_{+}$, where $N = \alpha^2$, and $\sinh^2(s) = n N_a$.
In this case, the $n$ probes are correlated, and for the observation we have
\begin{align}
& \bm{\mu}_{\mathrm{e}}(\eta) = \sqrt{N\eta} \cdot \bm{u}, \label{eqn:mu_ent} \\
& \bm{\Sigma}_{\mathrm{e}}(\eta) = \frac{1}{4} \bm{I} - \frac{\eta (1-e^{-2s})}{4n} \cdot \bm{u} \bm{u}^T. \label{eqn:Sigma_ent}
\end{align}
To compute the FI $\mathcal{I}(\eta)$, first note that $\bm{\Sigma}(\eta)$ is the sum of a diagonal matrix and a rank-$1$ matrix. 
This allows us obtain a closed-form for $\bm{\Sigma}(\eta)$ by applying the Sherman-Morrison formula~\cite{sherman1950adjustment}:
\begin{align}
\left(\bm{\Sigma}_{\mathrm{e}}(\eta)\right)^{-1}
= 4\bm{I} + \frac{4c_n \eta}{n(1-c_n \eta)} \bm{u} \bm{u}^T, \nonumber
\end{align}
where
\begin{equation} \label{eqn:c_n}
c_n = 1-e^{-2s} = \frac{2\sqrt{nN_a}}{\sqrt{nN_a+1}+\sqrt{nN_a}}.
\end{equation}
The notation $c_n$ will recur throughout this paper.
While all results can be given in terms of the original parameters of the probes (i.e., $N, N_a, n$), the $c_n$ notation shortens the expressions, which makes them more intuitive to reason about.
We will frequently make use of the following properties of $c_n$: 
\begin{enumerate*}[label=(\roman*)]
\item $c_n \in (0,1)$, and 
\item for any fixed $N_a$, $c_n$ increases in $n$ and approaches $1$ as $n \to \infty$.
\end{enumerate*}
Then~\eqref{eqn:FIM} gives the FI $I_{\mathrm{e}}(\eta)$ for the entanglement case 
\begin{align}
I_{\mathrm{e}}(\eta)
= \frac{nN}{\eta(1-c_n \eta)} + \frac{c_n^2}{2(1-c_n \eta)^2}. \label{eqn:I_e}
\end{align}

\paragraph{Quantum vs. classical.}
Throughout this paper, we refer to probes implemented by coherent-state pulses as \emph{classical} probes, and squeezing- and entanglement-augmented probes as \emph{quantum} probes.
To see how quantum probes fare relative to the classical ones, we compare their FIs, $I_{\mathrm{c}}(\eta)$, $I_{\mathrm{s}}(\eta)$ and $I_{\mathrm{e}}(\eta)$.
We first bring $I_{\mathrm{s}}(\eta)$ into the same parametrization as $I_{\mathrm{c}}(\eta)$ and $I_{\mathrm{e}}(\eta)$.
Note that by~\eqref{eqn:c_n}, $1-e^{-2r} = c_1$ and~\eqref{eqn:I_s_with_r} gives
\begin{equation} \label{eqn:I_s}
I_{\mathrm{s}}(\eta) = n \left( \frac{N}{\eta(1-c_1 \eta)} + \frac{c_1^2}{2(1-c_1 \eta)^2} \right).
\end{equation}
Next we show that, as long as the classical energy $N$ is large enough compared to the quantum energy $N_a$, entanglement augmentation is superior to squeezing (Fact~\ref{fact:single_ch_e_vs_s}), and quantum probes outperform their classical counterparts (Fact~\ref{fact:single_ch_q_vs_c}).
\begin{fact} \label{fact:single_ch_e_vs_s}
If $N > \frac{c_1^2}{2(c_n - c_1)}$, $I_{\mathrm{e}}(\eta) > I_{\mathrm{s}}(\eta)$.
\end{fact}
\begin{proof}
Consider $I_{\mathrm{e}}(\eta) - I_{\mathrm{s}}(\eta)$,
\begin{align*}
I_{\mathrm{e}}(\eta) - I_{\mathrm{s}}(\eta)
& = \frac{nN}{\eta}\left( \frac{1}{1-c_n \eta} - \frac{1}{1-c_1 \eta} \right) + \frac{1}{2} \left( \frac{c_n^2}{(1-c_n \eta)^2} - \frac{n c_1^2}{(1-c_1 \eta)^2}\right) \\
& > \frac{nN(c_n - c_1)}{(1-c_n \eta)(1-c_1 \eta)} - \frac{n c_1^2}{2(1-c_1 \eta)^2} \\
& > \frac{nN(c_n - c_1)}{(1-c_1 \eta)^2} - \frac{n c_1^2}{2(1-c_1 \eta)^2}
&& (c_n > c_1 > 0) \\
& = \frac{n}{2(1-c_1 \eta)^2} \left(2N(c_n - c_1) - c_1^2 \right).
\end{align*}
To have $I_{\mathrm{e}}(\eta) > I_{\mathrm{s}}(\eta)$, it suffices to set $2N(c_n - c_1) > c_1^2$.
\end{proof}

\begin{fact} \label{fact:single_ch_q_vs_c}
If $N > \left( \frac{1}{c_n \eta} - 1 \right) N_a$, $I_{\mathrm{e}}(\eta) > I_{\mathrm{c}}(\eta)$.
In particular, $I_{\mathrm{s}}(\eta) > I_{\mathrm{c}}(\eta)$ if $N > \left( \frac{1}{c_1 \eta} - 1 \right) N_a$.
\end{fact}
\begin{proof}
Note that for large enough $N$, the first term in $I_{\mathrm{e}}(\eta)$~\eqref{eqn:I_e} dominates, therefore $I_{\mathrm{e}}(\eta) > \frac{nN}{\eta(1-c_n \eta)}$.
This gives a simple lower bound $\frac{I_{\mathrm{e}}(\eta)}{I_{\mathrm{c}}(\eta)} > \frac{N}{(N+N_a)(1-c_n \eta)}$.
The first bound $N > \left( \frac{1}{c_n \eta} - 1 \right) N_a$ follows from setting $\frac{N}{(N+N_a)(1-c_n \eta)} > 1$.
Similarly, $\frac{I_{\mathrm{s}}(\eta)}{I_{\mathrm{c}}(\eta)} > \frac{N}{(N+N_a)(1-c_1 \eta)}$, and the second part of the claim follows.
\end{proof}
Though we tend to opt for simplicity over tightness throughout this work, Facts~\ref{fact:single_ch_e_vs_s} and~\ref{fact:single_ch_q_vs_c} quantify what it means for the classical energy $N$ to be large enough, a notion that recurs in this section.
Moreover, Facts~\ref{fact:single_ch_e_vs_s} and \ref{fact:single_ch_q_vs_c} provide practical bounds for realistic parameters.
For example, with $10 \log_{10}(e^{2r}) = 6 \mathrm{ dB}$ of squeezing per-pulse, for $n=2$ probes, Fact~\ref{fact:single_ch_e_vs_s} implies that entanglement augmentation outperforms squeezing alone, given $N>3.03$ photon equivalent of classical coherent energy, and regardless of how large the underlying channel transmissivity $\eta$ is.
On the other hand, for$\eta = 0.8$, Fact~\ref{fact:single_ch_q_vs_c} states that squeezed states are still better than classical coherent states if we merely have $N > 0.4$.

\subsection{Sharing entanglement across channels} \label{sec:physical:multi}

When there are multiple channels with unknown transmissivities, we further have the option to \emph{split} the entangled probes \emph{spatially} across different channels.
An interesting question that arises, is whether the correlation among probes, introduced by entanglement, still brings benefits as in the single-channel case (\S~\ref{sec:physical:single}).
Next, through analyzing two special configurations, we give strong evidence that the answer to this question is likely to be in the negative.

To be concrete, for the first configuration, we think of the $n$ channels as forming a tree, with one internal node $u$ as the sender of all the probes, and each leaf $v_i$ as the receiver for the channel along edge $(u,v_i)$, $i \in [n]$.
Consider $n$ probes (pulses), where the $i$-th probe is sent through the channel along $(u,v_i)$, whose unknown transmissivity is denoted as $\eta_i$ for all $i \in [n]$.
Crucially, we assume the $\eta_i$s do not depend on each other.
Fix the classical and quantum energy per-pulse, $N$ and $N_a$, we compare two cases:
\begin{enumerate*}[(i)]
\item each probe is a displaced squeezed state; and 
\item the $n$ probes are entangled.
\end{enumerate*}
We show that, with respect to \emph{both} the determinant of the FIM and the trace of its inverse, entanglement is in fact worse than independent squeezing.

When we go beyond the channel case and consider probes in a general network, the walk specified by each probe corresponds to a channel, and the channels may share edges.
Consequently, their channel transmissivities may also share common factors.
For example, when two probe both traverse a link $e$, their channel transmissivities share the common factor $\eta_e$.
Accordingly, we must consider whether the statements above still apply in this case.
Unfortunately, analyzing this case for $n$ channels, even with a  relatively clean setup of the the probes, turns out to be quite heavy.
We therefore include numerical results supporting the statements above for a two-channel setup in Appendix~\ref{apd:2_channel}, and analyze the configuration of $n$ independent channels in the remainder of this section.

\paragraph{Independent squeezing.}
Let $\bm{v} = (\sqrt{\eta_1}, \sqrt{\eta_2}, \dots, \sqrt{\eta_n})$.
When we send a displaced squeezed state with parameters $N$ and $N_a$ through each channel, by~\eqref{eqn:mu_sq}, \eqref{eqn:Sigma_sq} and~\eqref{eqn:c_n} for $n=1$,
\begin{align*}
& \bm{\mu}_{\mathrm{s,d}}(\bm{\eta}) = \sqrt{N} \cdot \bm{v}, \\
& \bm{\Sigma}_{\mathrm{s,d}}(\bm{\eta}) = \frac{1}{4} \bm{I} - \frac{c_1}{4} \cdot \operatorname{diag}(\eta_1, \eta_2, \dots, \eta_n).
\end{align*}
Next we consider the FIM of this case, $\mathcal{I}_{\mathrm{s,d}}$.
Since each channel transmissivity $\eta_i$ has no dependence on $\eta_j$ for $j \neq i$, and the covariance matrix $\bm{\Sigma}_{\mathrm{s,d}}(\eta)$ is diagonal, it is easy to check that $\mathcal{I}_{\mathrm{s,d}}$ is also diagonal, where $(\mathcal{I}_{\mathrm{s,d}})_{i,i} = \frac{N}{\eta_i (1-c_1 \eta_i)} + \frac{c_1^2}{2(1-c_1 \eta_i)^2}$, following~\eqref{eqn:I_s} for $n=1$.
Then, 
\begin{align}
\det(\mathcal{I}_{\mathrm{s,d}}) 
& = \prod_{i=1}^n \left( \frac{N}{\eta_i (1-c_1 \eta_i)} + \frac{c_1^2}{2(1-c_1 \eta_i)^2} \right)  \nonumber \\
& = \frac{N^n}{\prod_{i=1}^n \eta_i} \cdot \prod_{i=1}^n \left( \frac{1}{1-c_1 \eta_i} + \frac{c_1^2 \eta_i}{2N(1-c_1 \eta_i)^2} \right) \label{eqn:det(I_{s,d})}, \\
\Tr(\mathcal{I}_{\mathrm{s,d}}^{-1})
& = \sum_{i=1}^n \frac{2\eta_i (1-c_1 \eta_i)^2}{2N(1-c_1 \eta_i) + c_1^2 \eta_i}. \nonumber
\end{align}

\paragraph{Correlated probes through entanglement.}
Here we split the block of $n$ entangled pulses across $n$ channels, where the $i$-th pulse is sent through the channel with transmissivity $\eta_i$, then we have
\begin{align}
& \bm{\mu}_{\mathrm{e,d}}(\bm{\eta}) = \sqrt{N} \cdot \bm{v}, \label{eqn:mu_ent_split} \\
& \bm{\Sigma}_{\mathrm{e,d}}(\bm{\eta}) = \frac{1}{4} \bm{I} - \frac{c_n}{4n} \cdot \bm{v} \bm{v}^T. \label{eqn:Sigma_ent_split}
\end{align}
In the special case of $\eta_1 = \eta_2 = \cdots = \eta_n = \eta$, $\bm{v} = \sqrt{\eta} \cdot \bm{u}$ and $\bm{v} \bm{v}^T = \eta \cdot \bm{u} \bm{u}^T$, \eqref{eqn:mu_ent_split} and~\eqref{eqn:Sigma_ent_split} indeed reduce to~\eqref{eqn:mu_ent} and~\eqref{eqn:Sigma_ent} respectively.
Next we shall see that the FIM for this case, denoted as $\mathcal{I}_{\mathrm{e,d}}$, is in fact the sum of a diagonal matrix and a scaling of the all-ones matrix.
\begin{claim} \label{clm:I_{e,d}}
$\mathcal{I}_{\mathrm{e,d}} = \beta \cdot \operatorname{diag}(\frac{1}{\eta_1}, \frac{1}{\eta_2}, \dots, \frac{1}{\eta_n}) + \gamma \bm{u} \bm{u}^T$, where 
$\bm{u}$ is the $n$-dimensional column vector of ones, 
$\beta = N + \frac{c_n^2 S_{\bm{\eta}}}{4n(n-c_n S_{\bm{\eta}})}$ and
$\gamma = \frac{c_nN}{n-c_nS_{\bm{\eta}}} + \frac{c_n^2(n+c_n S_{\bm{\eta}})}{4n(n-c_n S_{\bm{\eta}})^2}$, 
with $S_{\bm{\eta}} = \sum_{i=1}^n \eta_i$.
\end{claim}
The proof of Claim~\ref{clm:I_{e,d}}, given in Appendix~\ref{apd:proof_I_{e,d}}, involves careful yet straightforward derivations of each term in~\eqref{eqn:FIM}.
Claim~\ref{clm:I_{e,d}} shows that $\mathcal{I}_{\mathrm{e,d}}$ is yet again well structured, which allows us to derive $\det(\mathcal{I}_{\mathrm{e,d}})$ through a direct application of the matrix determinant lemma,
\begin{align}
\det(\mathcal{I}_{\mathrm{e,d}})
& = \frac{\beta^n}{\prod_{i=1}^n \eta_i} \cdot \left( 1+ \frac{1}{\beta}\bm{u}^T\cdot \operatorname{diag}(\eta_1, \eta_2, \dots, \eta_n) \cdot \gamma \bm{u} \right) \nonumber \\
& = \frac{\beta^n}{\prod_{i=1}^n \eta_i} \cdot \left( 1+ \frac{\gamma S_{\bm{\eta}}}{\beta}\right) \nonumber \\
& = \frac{2n}{(n-c_n S_{\bm{\eta}})} \cdot \frac{N^n}{\prod_{i=1}^n \eta_i} \cdot \left( 1 + \frac{c_n^2 S_{\bm{\eta}}}{4Nn(n-c_n S_{\bm{\eta}})} \right)^n \cdot 
\underbrace{\frac{2Nn(n-c_n S_{\bm{\eta}}) + c_n^2 S_{\bm{\eta}}}{4Nn(n-c_n S_{\bm{\eta}}) + c_n^2 S_{\bm{\eta}}}}_{(\spadesuit)}. \label{eqn:det(I_{e,d})}
\end{align}
By the Sherman-Morrison formula, we can also obtain $\mathcal{I}_{\mathrm{e,d}}^{-1}$, 
\begin{equation*}
\mathcal{I}_{\mathrm{e,d}}^{-1}
= \frac{1}{\beta} \cdot \operatorname{diag}(\frac{1}{\eta_1}, \frac{1}{\eta_2}, \dots, \frac{1}{\eta_n})
- \frac{\gamma}{\beta (\beta + \gamma S_{\bm{\eta}})} \cdot \bm{\eta} \bm{\eta}^T.
\end{equation*}
Denote $Q_{\bm{\eta}} = \sum_{i=1}^n \eta_i^2$, we have the trace of the inverse FIM as follows:
\begin{align*}
\Tr(\mathcal{I}_{\mathrm{e,d}}^{-1})
& = \frac{S_{\bm{\eta}}}{\beta} - \frac{\gamma Q_{\bm{\eta}}}{\beta (\beta + \gamma S_{\bm{\eta}})} \\
& = 2(n-c_n S_{\bm{\eta}}) \cdot \frac{4Nn(n-c_n S_{\bm{\eta}})(nS_{\bm{\eta}} - c_n Q_{\bm{\eta}}) + c_n^2(n Q_{\bm{\eta}} + c_nS_{\bm{\eta}}Q_{\bm{\eta}}) - 2nS_{\bm{\eta}}^2}{(4Nn(n-c_n S_{\bm{\eta}}) + c_n^2 S_{\bm{\eta}})(2Nn(n-c_n S_{\bm{\eta}}) + c_n^2 S_{\bm{\eta}})}.
\end{align*}

\subsubsection{Spatial entanglement induces a smaller determinant of the FIM}

Next, we show that independent squeezing is indeed better than entanglement, with respect to the determinant of the corresponding FIMs.
A formal proof would be inevitably cumbersome, so here we opt for an intuitive argument on why $\det(\mathcal{I}_{\mathrm{s,d}}) > \det(\mathcal{I}_{\mathrm{e,d}})$, as long as the classical coherent energy $N$ and the sum of the transmissivities $S_{\bm{\eta}}$ are not too small.

As is customary in proofs of such statements, we show that some lower bound on $\det(\mathcal{I}_{\mathrm{s,d}})$ \eqref{eqn:det(I_{s,d})} is still larger than some upper bound on $\det(\mathcal{I}_{\mathrm{e,d}})$ \eqref{eqn:det(I_{e,d})}.
In comparing these expressions, we ignore the common factor $\frac{N^n}{\prod_{i=1}^n \eta_i}$.
To obtain a suitable lower bound on $\det(\mathcal{I}_{\mathrm{s,d}})$, we apply Jensen's inequality to $\log\left( \frac{1}{1-c_1 x} + \frac{c_1^2 x}{2N(1-c_1 x)^2} \right)$, and get that $\prod_{i=1}^n \left( \frac{1}{1-c_1 \eta_i} + \frac{c_1^2 \eta_i}{2N(1-c_1 \eta_i)^2} \right)$ attains its minimum at $\eta_1 = \eta_2 = \cdots \eta_n = \frac{S_{\bm{\eta}}}{n}$.
Then define 
\begin{align}
f(S_{\bm{\eta}}) 
& = \prod_{i=1}^n \left( \frac{1}{1-c_1 \cdot \frac{S_{\bm{\eta}}}{n}} + \frac{c_1^2 \cdot \frac{S_{\bm{\eta}}}{n}}{2N(1-c_1 \cdot \frac{S_{\bm{\eta}}}{n})^2} \right) \nonumber \\
& = \left( \frac{n}{n-c_1 S_{\bm{\eta}}} \right)^n \left( 1 + \frac{c_1^2 S_{\bm{\eta}}}{2N(n-c_1 \cdot S_{\bm{\eta}})}  \right)^n. \label{eqn:f(S_eta)}
\end{align}
To upper bound $\det(\mathcal{I}_{\mathrm{e,d}})$, notice that for $N$ that is not too small, the $(\spadesuit)$ term in~\eqref{eqn:det(I_{e,d})} is very close to $0.5$, and can be upper bounded by an absolute constant of $0.55$.
Therefore, define
\begin{equation} \label{eqn:g(S_eta)}
g(S_{\bm{\eta}}) 
= \frac{1.1n}{n-c_n S_{\bm{\eta}}} \cdot \left( 1 + \frac{c_n^2 S_{\bm{\eta}}}{4Nn(n-c_n S_{\bm{\eta}})} \right)^n.
\end{equation}

$f(S_{\bm{\eta}}) > g(S_{\bm{\eta}})$ would be a sufficient condition for $\det(\mathcal{I}_{\mathrm{s,d}}) > \det(\mathcal{I}_{\mathrm{e,d}})$.
Note that for $N$ that is not too small, $\frac{c_1^2 S_{\bm{\eta}}}{2N(1-c_1 \cdot S_{\bm{\eta}})} > \frac{c_n^2 S_{\bm{\eta}}}{4Nn(n-c_n S_{\bm{\eta}})}$, so the second factor in~\eqref{eqn:f(S_eta)} is larger than that of~\eqref{eqn:g(S_eta)}.
Moreover, the first factor in~\eqref{eqn:f(S_eta)}, $\left( \frac{n}{n-c_1 S_{\bm{\eta}}} \right)^n$, grows more rapidly than $\frac{1.1n}{(n-c_n S_{\bm{\eta}})}$ in~\eqref{eqn:g(S_eta)}, albeit with a smaller initial value.
We hence conclude that given $N$ and $S_{\bm{\eta}}$ that are not too small, we have $f(S_{\bm{\eta}}) > g(S_{\bm{\eta}})$, and consequently $\det(\mathcal{I}_{\mathrm{s,d}}) > \det(\mathcal{I}_{\mathrm{e,d}})$.

As an example, given $N_a = 0.558$, which corresponds to having $6 \mathrm{dB}$ of squeezing per-pulse (see the end of \S~\ref{sec:physical:single}), if we send $n=2$ probes, we merely need the classical coherent energy $N \geq 6$ to guarantee that the $(\spadesuit)$ term in~\eqref{eqn:det(I_{e,d})} is bounded from above by $0.55$.
If we further impose that the combined transmissivities from the $n=2$ channels $S_{\bm{\eta}} \geq 0.26$, $f(S_{\bm{\eta}}) > g(S_{\bm{\eta}})$ holds, and we have $\det(\mathcal{I}_{\mathrm{s,d}}) > \det(\mathcal{I}_{\mathrm{e,d}})$.

\subsubsection{Spatial entanglement induces a larger trace of the inverse FIM}

\begin{figure}[t]
\centering
\includegraphics[width=0.6\textwidth]{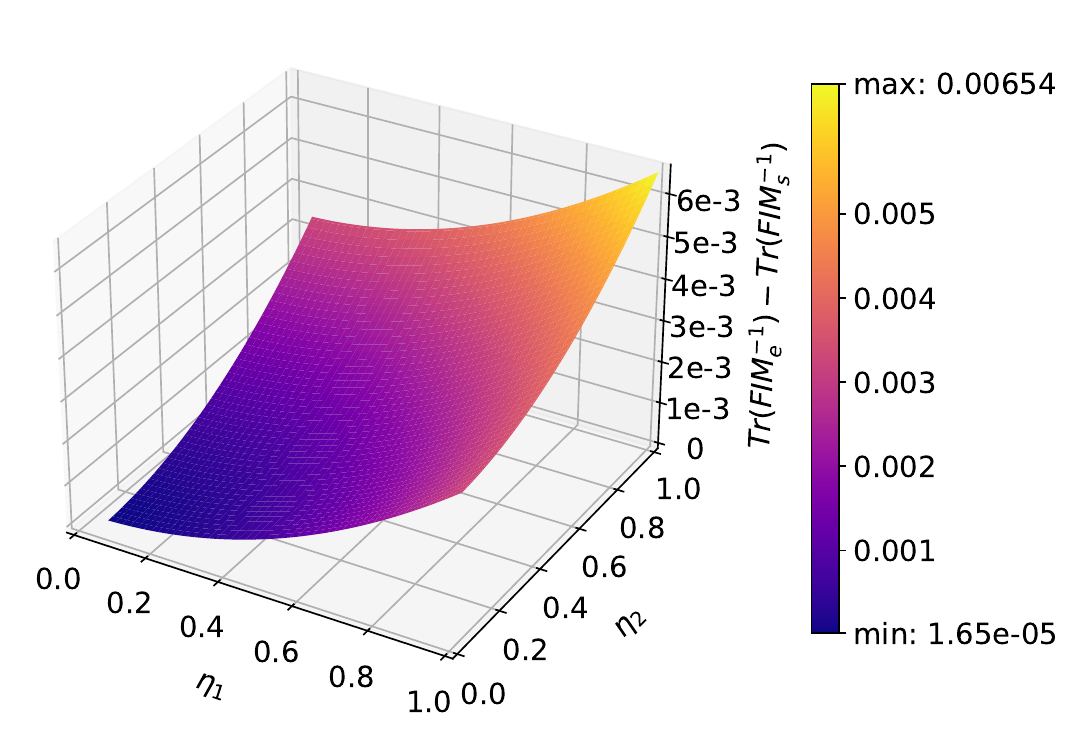}
\caption{$\Tr(\mathcal{I}_{\mathrm{e,d}}^{-1}) - \Tr(\mathcal{I}_{\mathrm{s,d}}^{-1})$ with respect to $\eta_1$ and $\eta_2$.}
\label{fig:n_indep_channels_tr}
\end{figure}

Finally, we sketch the argument for $\Tr(\mathcal{I}_{\mathrm{s,d}}^{-1}) < \Tr(\mathcal{I}_{\mathrm{e,d}}^{-1})$.
To upper bound $\Tr(\mathcal{I}_{\mathrm{s,d}}^{-1})$, note that the function $\frac{2\eta(1-c_1\eta)^2}{2N(1-c_1\eta) + c_1^2 \eta}$ is now concave in $\eta$, so Jensen's inequality conveniently gives that $\Tr(\mathcal{I}_{\mathrm{s,d}}^{-1})$ attains its maximum at $\eta_1 = \eta_2 = \cdots = \eta_n = \frac{S_{\bm{\eta}}}{n}$:
\begin{equation*}
\Tr(\mathcal{I}_{\mathrm{s,d}}^{-1})
\leq \frac{2S_{\bm{\eta}}(n-c_1S_{\bm{\eta}})^2}{2Nn(n-c_1S_{\bm{\eta}}) + nc_1^2 S_{\bm{\eta}}}.
\end{equation*}
When $N$ is large enough, $2Nn(n-c_1S_{\bm{\eta}})$ dominates $nc_1^2 S_{\bm{\eta}}$, we drop the $nc_1^2 S_{\bm{\eta}}$ term in the denominator to get a cleaner upper bound,
\begin{equation} \label{ineq:UB_Tr_sq}
\Tr(\mathcal{I}_{\mathrm{s,d}}^{-1})
\leq \frac{S_{\bm{\eta}}(n-c_1S_{\bm{\eta}})}{Nn}.
\end{equation}
Similarly for $\Tr(\mathcal{I}_{\mathrm{e,d}}^{-1})$, the intuition is best understood by focusing on the dominating terms, even though this does not yield a lower bound:
\begin{align}
\Tr(\mathcal{I}_{\mathrm{e,d}}^{-1})
& \sim 2(n-c_n S_{\bm{\eta}}) \cdot \frac{4Nn(n-c_n S_{\bm{\eta}})(nS_{\bm{\eta}} - c_n Q_{\bm{\eta}})}{4Nn(n-c_n S_{\bm{\eta}}) \cdot 2Nn(n-c_n S_{\bm{\eta}})} \nonumber \\
& = \frac{nS_{\bm{\eta}} - c_n Q_{\bm{\eta}}}{Nn}. \label{asym:Tr_ent}
\end{align}
For $\Tr(\mathcal{I}_{\mathrm{s,d}}^{-1}) < \Tr(\mathcal{I}_{\mathrm{e,d}}^{-1})$ to hold, by comparing~\eqref{ineq:UB_Tr_sq} and~\eqref{asym:Tr_ent}, we approximately need 
\begin{equation} \label{ineq:condition_Tr}
c_1 S_{\bm{\eta}}^2 > c_n Q_{\bm{\eta}}.
\end{equation}
Notice that~\eqref{ineq:condition_Tr} naturally holds if $S_{\bm{\eta}}$ is not too small.
This is because, $Q_{\bm{\eta}} < S_{\bm{\eta}}$ by definition, so when $S_{\bm{\eta}} \geq \frac{c_n}{c_1}$, we have $c_1 S_{\bm{\eta}}^2 \geq c_n S_{\bm{\eta}} > c_n Q_{\bm{\eta}}$.
When $S_{\bm{\eta}}$ is small, for~\eqref{ineq:condition_Tr} to hold, we need $\eta_1, \eta_2 \dots, \eta_n$ not too dispersed and concentrate around their mean $\frac{S_{\bm{\eta}}}{n}$.
The extent to which they need to be concentrated is governed by the quantum energy $N_a$ through the ratio $\frac{c_1}{c_n}$.
For a fixed $n$, as we increase $N_a$, $\frac{c_1}{c_n}$ increases towards $1$, and $\eta_1, \eta_2 \dots, \eta_n$ can be slightly more dispersed while satisfying~\eqref{ineq:condition_Tr}.

Although this argument provides intuition for the general case, it is not a tight analysis.
In the case of sending $n=2$ probes, with classical coherent energey per-pulse $N = 100$, quantum energy per-pulse $N_a = 0.558$ ($6 \mathrm{dB}$ of squeezing as before), numerical results indicate that $\Tr(\mathcal{I}_{\mathrm{s,d}}^{-1}) < \Tr(\mathcal{I}_{\mathrm{e,d}}^{-1})$ for all choices of $\eta_1$ and $\eta_2$ (Figure~\ref{fig:n_indep_channels_tr}).

%% file: algo.tex
\section{Routing probes in a network}
\label{sec:graph}

A network is an undirected graph $G=(V,E)$ with a given set of monitors $M \subseteq V$, where each edge $e \in E$ is associated with an unknown channel transmissivity $\eta_e \in (0,1]$. 
We think of each probe $P$ as a multiset of edges.
While distinct probes may intersect on a subset of common edges, what we showed in \S~\ref{sec:physical:multi} alludes to the fact that sharing entanglement among such probes is unlikely to be beneficial. 
This motivates us to again consider how each \emph{individual} probe is routed in the network, without worrying about the correlation among distinct probes, a setting similar to~\cite{zheng2025quantum}.

Here our goal is to construct a set of probes $\mathcal{P}$ that \emph{can} identify all link transmissivities, while ensuring the observations from $\mathcal{P}$ are as informationally orthogonal as possible.
We introduce the notions of \emph{identifiability} and \emph{information orthogonality} in \S~\ref{sec:graph:setup}, and connect information orthogonality with graph properties in \S~\ref{sec:graph:connect}.
Then in \S~\ref{sec:graph:algo}, we propose an algorithm that constructs a set of probes $\mathcal{P}$ with these desirable properties.
We illustrate the ideas from this section using a small network in Figure~\ref{fig:example}.

\subsection{The probe construction problem} \label{sec:graph:setup}

\paragraph{Identifiability.} In stochastic network tomography, a set of probes $\mathcal{P}$, together with its underlying network (graph) $G=(V,E)$, defines a $\sizeof{\mathcal{P}} \times \sizeof{E}$ measurement matrix $A$. 
Each row of $A$ corresponds to a probe $P \in \mathcal{P}$, each column in $A$ corresponds to an edge $e \in E$, and the $(i,j)$-entry $A_{i,j}$ describes the number of times the $i$-th probe $P_i$ passes through the $j$-th edge $e_j$. 
Using the measurement matrix $A$, the problem of estimating link transmissivities in a network can be cast as a linear system $A \bm{z} = \bm{w}$, where $\bm{z} = (\log \eta_1, \log \eta_2, \dots, \log \eta_{\sizeof{E}})$ is a bijection of $\bm{\eta}$, and $\bm{w} = (\log (\prod_{e \in P_1} \eta_e), \log (\prod_{e \in P_2} \eta_e), \dots, \log (\prod_{e \in P_{\sizeof{\mathcal{P}}}} \eta_e))$ is a column vector related to the transmissivities of each probe in $\mathcal{P}$. 
Since the observations from $\mathcal{P}$ depend on $\bm{\eta}$ only through $\bm{w}$, and $\bm{w}$ can be estimated consistently from $\mathcal{P}$, it is known that $\bm{\eta}$ is \emph{identifiable} if and only if $A$ has full column rank~\cite{he2015fisher}.
This indicates, that if we think of each column in $A$ as the \emph{signature} of that edge, for probes $\mathcal{P}$ to be able to identify \emph{all} link transmissivities, the signatures of the edges must be linearly independent.

\paragraph{Information orthogonality.}
In classical network tomography, once given a measurement matrix, and some experiment design, the problem of solving the unknown parameters is often treated as a maximum likelihood estimation problem~\cite{castro2004network,he2015fisher,xi2006estimating}.
In our context, if we use quantum implementations of the probes (\S~\ref{sec:physical:single}), solving the maximum likelihood estimator (MLE) of $\bm{\eta}$ requires solving a system of nonlinear equations, which generally admits no analytical solution.
Then, to obtain an approximation of the MLE, one has to rely on various gradient-based optimization methods, such as the BFGS algorithm~\cite{broyden1970convergence,fletcher1970new,goldfarb1970family,shanno1970conditioning}.

Recall that $\sizeof{E}$ is the number of links in a network.
If we send $w$ probes in total, and on average each probe passes through $t$ edges, each iteration of the BFGS algorithm has time complexity $O(\sizeof{E}^2+wt)$, and existing heuristics suggest that BFGS often requires roughly $O(\sizeof{E})$ iterations before converging rapidly~\cite{nocedal2006numerical}.
For a large network, a time complexity of $\Omega(\sizeof{E}^3)$ can be prohibitively expensive. 
In this case, if the probes are designed in a way that enables solving for smaller subnetworks in parallel, the time needed to solve for a large network can be greatly reduced.
This brings about the notion of information orthogonality.

\begin{definition}[Information orthogonality~\cite{cox1987parameter}]
$k$ sets of parameters $H_1, H_2, \dots, H_k$ are \emph{information orthogonal}, if for any $\eta_i \in H_x$ and any $\eta_j \in H_y$ with distinct $x, y \in [k]$, the corresponding entry $\mathcal{I}_{i,j}$ in the FIM $\mathcal{I}$ satisfies $\mathcal{I}_{i,j} = 0$.
\end{definition}

Specific to our case of network tomography, information orthogonality implies optimization decoupling.\footnote{This is not the case in general. However, in our context, the intuition is that, probes in one subnetwork does not contribute to objective functions (e.g., the maximum likelihood function) of other networks, therefore the objective functions of separate subnetworks can be decoupled, and information orthogonality implies optimization decoupling.}
This is to say, that the more information orthogonal sets of parameters we have, the more instances can be run in parallel.
Therefore, in constructing the probes, we solve the following problem:
\begin{problem} \label{prob}
Given a graph $G$ and a set of monitors $M$, construct a set of probes $\mathcal{P}$ that 
\begin{enumerate*}[label=(\roman*)]
\item can identify all link transmissivities; and 
\item maximizes the number of information orthogonal subsets of link transmissivities.
\end{enumerate*}
\end{problem}

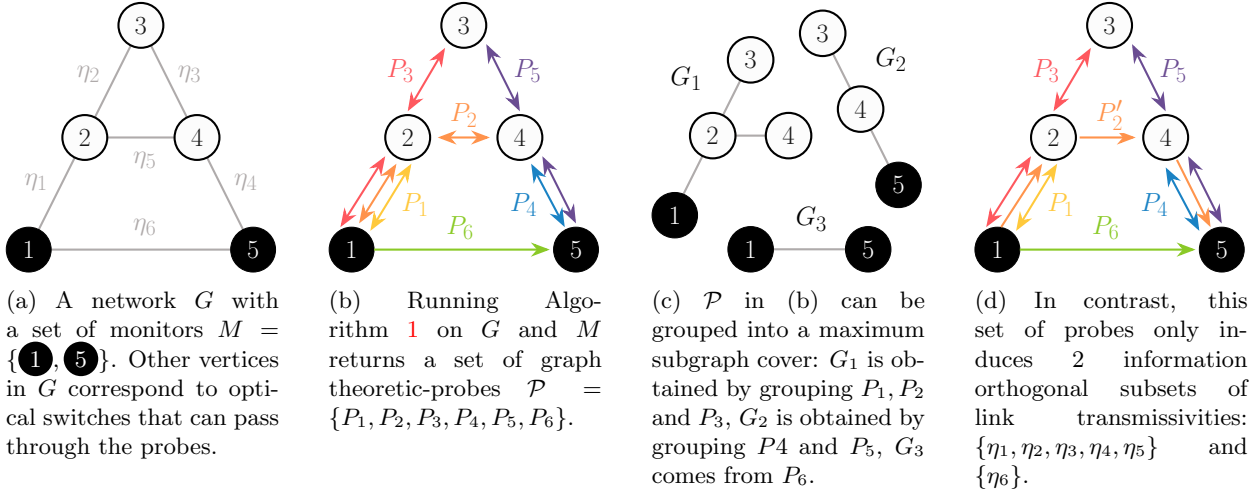
\begin{figure}
\centering
\begin{subfigure}[t]{0.22\textwidth}
\centering
\resizebox{1\textwidth}{!}{%
\begin{circuitikz}
\tikzstyle{every node}=[font=\fontsize{27.9pt}{36.2pt}\selectfont]
\node [font=\fontsize{27.9pt}{36.2pt}\selectfont, color={rgb,255:red,175; green,170; blue,170}, text opacity=1, , inner xsep=0.080cm, inner ysep=0.085cm, rounded corners=0.020cm] at (11.5,9.5) {$\eta_1$};
\draw [ color={rgb,255:red,175; green,171; blue,171}, draw opacity=1, line width=2pt, short] (11.25,7.25) -- (18.75,7.25);
\draw [ color={rgb,255:red,175; green,171; blue,171}, draw opacity=1, line width=2pt, short] (11.25,7.25) -- (15,14.75);
\draw [ color={rgb,255:red,175; green,171; blue,171}, draw opacity=1, line width=2pt, short] (15,14.75) -- (18.75,7.25);
\draw [ color={rgb,255:red,175; green,171; blue,171}, draw opacity=1, line width=2pt, short] (13.125,11) -- (16.875,11);
\draw [ fill={rgb,255:red,0; green,0; blue,0}, fill opacity=1] (11.25,7.25) circle (0.75cm);
\node [font=\fontsize{27.9pt}{36.2pt}\selectfont, color={rgb,255:red,255; green,255; blue,255}, text opacity=1, , inner xsep=0.080cm, inner ysep=0.085cm, rounded corners=0.020cm] at (11.25,7.25) {$1$};
\draw [ fill={rgb,255:red,0; green,0; blue,0}, fill opacity=1] (18.75,7.25) circle (0.75cm);
\node [font=\fontsize{27.9pt}{36.2pt}\selectfont, color={rgb,255:red,255; green,255; blue,255}, text opacity=1, , inner xsep=0.080cm, inner ysep=0.085cm, rounded corners=0.020cm] at (18.75,7.25) {$5$};
\draw [ fill={rgb,255:red,252; green,252; blue,252}, fill opacity=1, line width=2pt ] (13.125,11) circle (0.75cm);
\draw [ fill={rgb,255:red,252; green,252; blue,252}, fill opacity=1, line width=2pt ] (16.875,11) circle (0.75cm);
\draw [ fill={rgb,255:red,252; green,252; blue,252}, fill opacity=1, line width=2pt ] (15,14.75) circle (0.75cm);
\node [font=\fontsize{27.9pt}{36.2pt}\selectfont, color={rgb,255:red,23; green,23; blue,23}, text opacity=1, , inner xsep=0.080cm, inner ysep=0.085cm, rounded corners=0.020cm] at (15,14.75) {$3$};
\node [font=\fontsize{27.9pt}{36.2pt}\selectfont, color={rgb,255:red,23; green,23; blue,23}, text opacity=1, , inner xsep=0.080cm, inner ysep=0.085cm, rounded corners=0.020cm] at (16.875,11) {$4$};
\node [font=\fontsize{27.9pt}{36.2pt}\selectfont, color={rgb,255:red,23; green,23; blue,23}, text opacity=1, , inner xsep=0.080cm, inner ysep=0.085cm, rounded corners=0.020cm] at (13.125,11) {$2$};
\node [font=\fontsize{27.9pt}{36.2pt}\selectfont, color={rgb,255:red,175; green,170; blue,170}, text opacity=1, , inner xsep=0.080cm, inner ysep=0.085cm, rounded corners=0.020cm] at (13.25,13.125) {$\eta_2$};
\node [font=\fontsize{27.9pt}{36.2pt}\selectfont, color={rgb,255:red,175; green,170; blue,170}, text opacity=1, , inner xsep=0.080cm, inner ysep=0.085cm, rounded corners=0.020cm] at (16.625,13.125) {$\eta_3$};
\node [font=\fontsize{27.9pt}{36.2pt}\selectfont, color={rgb,255:red,175; green,170; blue,170}, text opacity=1, , inner xsep=0.080cm, inner ysep=0.085cm, rounded corners=0.020cm] at (18.5,9.5) {$\eta_4$};
\node [font=\fontsize{27.9pt}{36.2pt}\selectfont, color={rgb,255:red,175; green,170; blue,170}, text opacity=1, , inner xsep=0.080cm, inner ysep=0.085cm, rounded corners=0.020cm] at (15.125,10.25) {$\eta_5$};
\node [font=\fontsize{27.9pt}{36.2pt}\selectfont, color={rgb,255:red,175; green,170; blue,170}, text opacity=1, , inner xsep=0.080cm, inner ysep=0.085cm, rounded corners=0.020cm] at (15.125,8) {$\eta_6$};
\end{circuitikz}
}%
\caption{A network $G$ with a set of monitors $M = \{\circled{1}, \circled{5}\}$. Other vertices in $G$ correspond to optical switches that can pass through the probes.}
\label{fig:example:G}
\end{subfigure}
\hfill
\begin{subfigure}[t]{0.22\textwidth}
\centering
\resizebox{1\textwidth}{!}{%
\begin{circuitikz}
\tikzstyle{every node}=[font=\fontsize{27.9pt}{36.2pt}\selectfont]
\draw [ fill={rgb,255:red,0; green,0; blue,0}, fill opacity=1] (11.25,7.25) circle (0.75cm);
\node [font=\fontsize{27.9pt}{36.2pt}\selectfont, color={rgb,255:red,255; green,255; blue,255}, text opacity=1, , inner xsep=0.080cm, inner ysep=0.085cm, rounded corners=0.020cm] at (11.25,7.25) {$1$};
\draw [ fill={rgb,255:red,0; green,0; blue,0}, fill opacity=1] (18.75,7.25) circle (0.75cm);
\node [font=\fontsize{27.9pt}{36.2pt}\selectfont, color={rgb,255:red,255; green,255; blue,255}, text opacity=1, , inner xsep=0.080cm, inner ysep=0.085cm, rounded corners=0.020cm] at (18.75,7.25) {$5$};
\draw [ fill={rgb,255:red,252; green,252; blue,252}, fill opacity=1, line width=2pt ] (13.125,11) circle (0.75cm);
\draw [ fill={rgb,255:red,252; green,252; blue,252}, fill opacity=1, line width=2pt ] (16.875,11) circle (0.75cm);
\draw [ fill={rgb,255:red,252; green,252; blue,252}, fill opacity=1, line width=2pt ] (15,14.75) circle (0.75cm);
\node [font=\fontsize{27.9pt}{36.2pt}\selectfont, color={rgb,255:red,23; green,23; blue,23}, text opacity=1, , inner xsep=0.080cm, inner ysep=0.085cm, rounded corners=0.020cm] at (15,14.75) {$3$};
\node [font=\fontsize{27.9pt}{36.2pt}\selectfont, color={rgb,255:red,23; green,23; blue,23}, text opacity=1, , inner xsep=0.080cm, inner ysep=0.085cm, rounded corners=0.020cm] at (16.875,11) {$4$};
\node [font=\fontsize{27.9pt}{36.2pt}\selectfont, color={rgb,255:red,23; green,23; blue,23}, text opacity=1, , inner xsep=0.080cm, inner ysep=0.085cm, rounded corners=0.020cm] at (13.125,11) {$2$};
\draw [ color={rgb,255:red,138; green,201; blue,39}, draw opacity=1, line width=2pt, -{Stealth[scale=1.5]}, ] (12,7.25) -- (17.875,7.25);
\draw [ color={rgb,255:red,255; green,147; blue,76}, draw opacity=1, line width=2pt, {Stealth[scale=1.5]}-{Stealth[scale=1.5]}, ] (11.5,8.125) -- (12.75,10.125);
\draw [ color={rgb,255:red,254; green,88; blue,94}, draw opacity=1, line width=2pt, {Stealth[scale=1.5]}-{Stealth[scale=1.5]}, ] (10.875,8) -- (12.375,10.375);
\draw [ color={rgb,255:red,254; green,88; blue,94}, draw opacity=1, line width=2pt, {Stealth[scale=1.5]}-{Stealth[scale=1.5]}, ] (13.125,11.875) -- (14.375,14);
\draw [ color={rgb,255:red,255; green,201; blue,60}, draw opacity=1, line width=2pt, {Stealth[scale=1.5]}-{Stealth[scale=1.5]}, ] (11.875,7.875) -- (13.25,10.125);
\draw [ color={rgb,255:red,106; green,76; blue,146}, draw opacity=1, line width=2pt, {Stealth[scale=1.5]}-{Stealth[scale=1.5]}, ] (18.75,8.125) -- (17.625,10.375);
\draw [ color={rgb,255:red,106; green,76; blue,146}, draw opacity=1, line width=2pt, {Stealth[scale=1.5]}-{Stealth[scale=1.5]}, ] (16.875,11.875) -- (15.75,14.125);
\draw [ color={rgb,255:red,255; green,147; blue,76}, draw opacity=1, line width=2pt, {Stealth[scale=1.5]}-{Stealth[scale=1.5]}, ] (14.125,11) -- (15.875,11);
\draw [ color={rgb,255:red,26; green,130; blue,196}, draw opacity=1, line width=2pt, {Stealth[scale=1.5]}-{Stealth[scale=1.5]}, ] (18.375,8) -- (17.25,10.125);
\node [font=\fontsize{27.9pt}{36.2pt}\selectfont, color={rgb,255:red,255; green,200; blue,60}, text opacity=1, , fill={rgb,255:red,255; green,255; blue,255}, fill opacity=1, text opacity=1, inner xsep=0.080cm, inner ysep=0.085cm, rounded corners=0.020cm] at (13.375,8.75) {$P_1$};
\node [font=\fontsize{27.9pt}{36.2pt}\selectfont, color={rgb,255:red,253; green,87; blue,94}, text opacity=1, , fill={rgb,255:red,255; green,255; blue,255}, fill opacity=1, text opacity=1, inner xsep=0.080cm, inner ysep=0.085cm, rounded corners=0.020cm] at (12.875,13.25) {$P_3$};
\node [font=\fontsize{27.9pt}{36.2pt}\selectfont, color={rgb,255:red,255; green,148; blue,76}, text opacity=1, , fill={rgb,255:red,255; green,255; blue,255}, fill opacity=1, text opacity=1, inner xsep=0.080cm, inner ysep=0.085cm, rounded corners=0.020cm] at (15,11.75) {$P_2$};
\node [font=\fontsize{27.9pt}{36.2pt}\selectfont, color={rgb,255:red,107; green,76; blue,146}, text opacity=1, , fill={rgb,255:red,255; green,255; blue,255}, fill opacity=1, text opacity=1, inner xsep=0.080cm, inner ysep=0.085cm, rounded corners=0.020cm] at (17.125,13.25) {$P_5$};
\node [font=\fontsize{27.9pt}{36.2pt}\selectfont, color={rgb,255:red,26; green,130; blue,196}, text opacity=1, , fill={rgb,255:red,255; green,255; blue,255}, fill opacity=1, text opacity=1, inner xsep=0.080cm, inner ysep=0.085cm, rounded corners=0.020cm] at (17,8.75) {$P_4$};
\node [font=\fontsize{27.9pt}{36.2pt}\selectfont, color={rgb,255:red,138; green,202; blue,38}, text opacity=1, , fill={rgb,255:red,255; green,255; blue,255}, fill opacity=1, text opacity=1, inner xsep=0.080cm, inner ysep=0.085cm, rounded corners=0.020cm] at (14.875,8) {$P_6$};
\end{circuitikz}
}%
\caption{Running Algorithm~\ref{algo:FindProbe} on $G$ and $M$ returns a set of graph theoretic-probes $\mathcal{P} = \{P_1, P_2, P_3, P_4, P_5, P_6\}$.}
\label{fig:example:output}
\end{subfigure}
\hfill
\begin{subfigure}[t]{0.22\textwidth}
\centering
\resizebox{1\textwidth}{!}{%
\begin{circuitikz}
\tikzstyle{every node}=[font=\fontsize{27.9pt}{36.2pt}\selectfont]
\draw [ color={rgb,255:red,175; green,171; blue,171}, draw opacity=1, line width=2pt, short] (16.125,5.625) -- (18.625,0.625);
\draw [ color={rgb,255:red,175; green,171; blue,171}, draw opacity=1, line width=2pt, short] (11.25,-0.25) -- (13.75,4.75);
\draw [ color={rgb,255:red,175; green,171; blue,171}, draw opacity=1, line width=2pt, short] (12.5,2.25) -- (15,2.25);
\draw [ fill={rgb,255:red,0; green,0; blue,0}, fill opacity=1] (11.25,-0.375) circle (0.75cm);
\node [font=\fontsize{27.9pt}{36.2pt}\selectfont, color={rgb,255:red,255; green,255; blue,255}, text opacity=1, , inner xsep=0.080cm, inner ysep=0.085cm, rounded corners=0.020cm] at (11.25,-0.375) {$1$};
\draw [ fill={rgb,255:red,0; green,0; blue,0}, fill opacity=1] (18.625,0.625) circle (0.75cm);
\node [font=\fontsize{27.9pt}{36.2pt}\selectfont, color={rgb,255:red,255; green,255; blue,255}, text opacity=1, , inner xsep=0.080cm, inner ysep=0.085cm, rounded corners=0.020cm] at (18.625,0.625) {$5$};
\draw [ fill={rgb,255:red,252; green,252; blue,252}, fill opacity=1, line width=2pt ] (12.5,2.25) circle (0.75cm);
\draw [ fill={rgb,255:red,252; green,252; blue,252}, fill opacity=1, line width=2pt ] (17.375,3.125) circle (0.75cm);
\draw [ fill={rgb,255:red,252; green,252; blue,252}, fill opacity=1, line width=2pt ] (16.125,5.625) circle (0.75cm);
\node [font=\fontsize{27.9pt}{36.2pt}\selectfont, color={rgb,255:red,23; green,23; blue,23}, text opacity=1, , inner xsep=0.080cm, inner ysep=0.085cm, rounded corners=0.020cm] at (16.125,5.625) {$3$};
\node [font=\fontsize{27.9pt}{36.2pt}\selectfont, color={rgb,255:red,23; green,23; blue,23}, text opacity=1, , inner xsep=0.080cm, inner ysep=0.085cm, rounded corners=0.020cm] at (17.375,3.125) {$4$};
\node [font=\fontsize{27.9pt}{36.2pt}\selectfont, color={rgb,255:red,23; green,23; blue,23}, text opacity=1, , inner xsep=0.080cm, inner ysep=0.085cm, rounded corners=0.020cm] at (12.5,2.25) {$2$};
\draw [ color={rgb,255:red,175; green,171; blue,171}, draw opacity=1, line width=2pt, short] (13.875,-1.5) -- (17.625,-1.5);
\draw [ fill={rgb,255:red,0; green,0; blue,0}, fill opacity=1] (17.625,-1.5) circle (0.75cm);
\node [font=\fontsize{27.9pt}{36.2pt}\selectfont, color={rgb,255:red,255; green,255; blue,255}, text opacity=1, , inner xsep=0.080cm, inner ysep=0.085cm, rounded corners=0.020cm] at (17.625,-1.5) {$5$};
\draw [ fill={rgb,255:red,0; green,0; blue,0}, fill opacity=1] (13.75,-1.5) circle (0.75cm);
\node [font=\fontsize{27.9pt}{36.2pt}\selectfont, color={rgb,255:red,255; green,255; blue,255}, text opacity=1, , inner xsep=0.080cm, inner ysep=0.085cm, rounded corners=0.020cm] at (13.75,-1.5) {$1$};
\draw [ fill={rgb,255:red,252; green,252; blue,252}, fill opacity=1, line width=2pt ] (13.75,4.75) circle (0.75cm);
\node [font=\fontsize{27.9pt}{36.2pt}\selectfont, color={rgb,255:red,23; green,23; blue,23}, text opacity=1, , inner xsep=0.080cm, inner ysep=0.085cm, rounded corners=0.020cm] at (13.75,4.75) {$3$};
\draw [ fill={rgb,255:red,252; green,252; blue,252}, fill opacity=1, line width=2pt ] (15,2.25) circle (0.75cm);
\node [font=\fontsize{27.9pt}{36.2pt}\selectfont, color={rgb,255:red,23; green,23; blue,23}, text opacity=1, , inner xsep=0.080cm, inner ysep=0.085cm, rounded corners=0.020cm] at (15,2.25) {$4$};
\node [font=\fontsize{27.9pt}{36.2pt}\selectfont, inner xsep=0.080cm, inner ysep=0.085cm, rounded corners=0.020cm] at (11.625,4.125) {$G_1$};
\node [font=\fontsize{27.9pt}{36.2pt}\selectfont, inner xsep=0.080cm, inner ysep=0.085cm, rounded corners=0.020cm] at (18.375,4.75) {$G_2$};
\node [font=\fontsize{27.9pt}{36.2pt}\selectfont, inner xsep=0.080cm, inner ysep=0.085cm, rounded corners=0.020cm] at (15.75,-0.5) {$G_3$};
\end{circuitikz}
}%
\caption{$\mathcal{P}$ in (b) can be grouped into a maximum subgraph cover: $G_1$ is obtained by grouping $P_1, P_2$ and $P_3$, $G_2$ is obtained by grouping $P4$ and $P_5$, $G_3$ comes from $P_6$.}
\label{fig:example:cover}
\end{subfigure}
\hfill
\begin{subfigure}[t]{0.22\textwidth}
\centering
\resizebox{1\textwidth}{!}{%
\begin{circuitikz}
\tikzstyle{every node}=[font=\fontsize{27.9pt}{36.2pt}\selectfont]
\draw [ fill={rgb,255:red,0; green,0; blue,0}, fill opacity=1] (11.25,7.75) circle (0.75cm);
\node [font=\fontsize{27.9pt}{36.2pt}\selectfont, color={rgb,255:red,255; green,255; blue,255}, text opacity=1, , inner xsep=0.080cm, inner ysep=0.085cm, rounded corners=0.020cm] at (11.25,7.75) {$1$};
\draw [ fill={rgb,255:red,0; green,0; blue,0}, fill opacity=1] (18.75,7.75) circle (0.75cm);
\node [font=\fontsize{27.9pt}{36.2pt}\selectfont, color={rgb,255:red,255; green,255; blue,255}, text opacity=1, , inner xsep=0.080cm, inner ysep=0.085cm, rounded corners=0.020cm] at (18.75,7.75) {$5$};
\draw [ fill={rgb,255:red,252; green,252; blue,252}, fill opacity=1, line width=2pt ] (13.125,11.5) circle (0.75cm);
\draw [ fill={rgb,255:red,252; green,252; blue,252}, fill opacity=1, line width=2pt ] (16.875,11.5) circle (0.75cm);
\draw [ fill={rgb,255:red,252; green,252; blue,252}, fill opacity=1, line width=2pt ] (15,15.25) circle (0.75cm);
\node [font=\fontsize{27.9pt}{36.2pt}\selectfont, color={rgb,255:red,23; green,23; blue,23}, text opacity=1, , inner xsep=0.080cm, inner ysep=0.085cm, rounded corners=0.020cm] at (15,15.25) {$3$};
\node [font=\fontsize{27.9pt}{36.2pt}\selectfont, color={rgb,255:red,23; green,23; blue,23}, text opacity=1, , inner xsep=0.080cm, inner ysep=0.085cm, rounded corners=0.020cm] at (16.875,11.5) {$4$};
\node [font=\fontsize{27.9pt}{36.2pt}\selectfont, color={rgb,255:red,23; green,23; blue,23}, text opacity=1, , inner xsep=0.080cm, inner ysep=0.085cm, rounded corners=0.020cm] at (13.125,11.5) {$2$};
\draw [ color={rgb,255:red,138; green,201; blue,39}, draw opacity=1, line width=2pt, -{Stealth[scale=1.5]}, ] (12,7.75) -- (17.875,7.75);
\draw [ color={rgb,255:red,255; green,148; blue,76}, draw opacity=1, line width=2pt, -{Stealth[scale=1.5]}, ] (11.375,8.5) -- (12.75,10.625);
\draw [ color={rgb,255:red,254; green,88; blue,94}, draw opacity=1, line width=2pt, {Stealth[scale=1.5]}-{Stealth[scale=1.5]}, ] (10.875,8.5) -- (12.375,10.875);
\draw [ color={rgb,255:red,254; green,88; blue,94}, draw opacity=1, line width=2pt, {Stealth[scale=1.5]}-{Stealth[scale=1.5]}, ] (13.125,12.375) -- (14.375,14.5);
\draw [ color={rgb,255:red,255; green,201; blue,60}, draw opacity=1, line width=2pt, {Stealth[scale=1.5]}-{Stealth[scale=1.5]}, ] (11.875,8.375) -- (13.25,10.625);
\draw [ color={rgb,255:red,106; green,76; blue,146}, draw opacity=1, line width=2pt, {Stealth[scale=1.5]}-{Stealth[scale=1.5]}, ] (18.75,8.625) -- (17.625,10.875);
\draw [ color={rgb,255:red,106; green,76; blue,146}, draw opacity=1, line width=2pt, {Stealth[scale=1.5]}-{Stealth[scale=1.5]}, ] (16.875,12.375) -- (15.75,14.625);
\draw [ color={rgb,255:red,255; green,147; blue,76}, draw opacity=1, line width=2pt, -{Stealth[scale=1.5]}, ] (14,11.5) -- (15.875,11.5);
\draw [ color={rgb,255:red,26; green,130; blue,196}, draw opacity=1, line width=2pt, {Stealth[scale=1.5]}-{Stealth[scale=1.5]}, ] (18,8.375) -- (16.875,10.5);
\node [font=\fontsize{27.9pt}{36.2pt}\selectfont, color={rgb,255:red,255; green,200; blue,60}, text opacity=1, , fill={rgb,255:red,255; green,255; blue,255}, fill opacity=1, text opacity=1, inner xsep=0.080cm, inner ysep=0.085cm, rounded corners=0.020cm] at (13.375,9.25) {$P_1$};
\node [font=\fontsize{27.9pt}{36.2pt}\selectfont, color={rgb,255:red,253; green,87; blue,94}, text opacity=1, , fill={rgb,255:red,255; green,255; blue,255}, fill opacity=1, text opacity=1, inner xsep=0.080cm, inner ysep=0.085cm, rounded corners=0.020cm] at (12.875,13.75) {$P_3$};
\node [font=\fontsize{27.9pt}{36.2pt}\selectfont, color={rgb,255:red,255; green,148; blue,76}, text opacity=1, , fill={rgb,255:red,255; green,255; blue,255}, fill opacity=1, text opacity=1, inner xsep=0.080cm, inner ysep=0.085cm, rounded corners=0.020cm] at (15,12.25) {$P'_2$};
\node [font=\fontsize{27.9pt}{36.2pt}\selectfont, color={rgb,255:red,107; green,76; blue,146}, text opacity=1, , fill={rgb,255:red,255; green,255; blue,255}, fill opacity=1, text opacity=1, inner xsep=0.080cm, inner ysep=0.085cm, rounded corners=0.020cm] at (17.125,13.75) {$P_5$};
\node [font=\fontsize{27.9pt}{36.2pt}\selectfont, color={rgb,255:red,26; green,130; blue,196}, text opacity=1, , fill={rgb,255:red,255; green,255; blue,255}, fill opacity=1, text opacity=1, inner xsep=0.080cm, inner ysep=0.085cm, rounded corners=0.020cm] at (16.5,9.25) {$P_4$};
\node [font=\fontsize{27.9pt}{36.2pt}\selectfont, color={rgb,255:red,138; green,202; blue,38}, text opacity=1, , fill={rgb,255:red,255; green,255; blue,255}, fill opacity=1, text opacity=1, inner xsep=0.080cm, inner ysep=0.085cm, rounded corners=0.020cm] at (14.875,8.5) {$P_6$};
\draw [ color={rgb,255:red,255; green,148; blue,76}, draw opacity=1, line width=2pt, -{Stealth[scale=1.5]}, ] (17.25,10.75) -- (18.375,8.5);
\end{circuitikz}
}%
\caption{In contrast, this set of probes only induces $2$ information orthogonal subsets of link transmissivities: $\{\eta_1, \eta_2, \eta_3, \eta_4, \eta_5\}$ and $\{\eta_6\}$.}
\label{fig:example:subopt}
\end{subfigure}
\caption{A network and two sets of probes that guarantees identifiability.}
\label{fig:example}
\end{figure}

\subsection{Connecting information orthogonality with graph properties} \label{sec:graph:connect}

Compared to the general multi-parameter estimation problems, the uniqueness of network tomography stems from the fact that the estimated parameters are graph-constrained.
Next, we establish connections between information orthogonality and graph properties, which would allow us to upper bound the maximum number of information orthogonal sets for a given network topology, and set the stage for proving the optimality of our probe construction algorithm in \S~\ref{sec:graph:algo}.

First we show that if two probes are edge-disjoint, their corresponding sets of link transmissivities are information orthogonal.
Recall that in this section, we think of a probe $P$ as a \emph{multiset} of edges.
Let $\supp(P)$ denote the underlying set of $P$, i.e., the set of distinct edges probe $P$ passes through.
\begin{claim} \label{clm:disjoint_probes}
For probes $P_1$ and $P_2$, let $H_i$ be the set of link transmissivities associated with edges in $P_i$, $H_i = \{ \eta_i \mid \forall e_i \in \supp(P_i)\}$, $i=1,2$.
$H_1$ and $H_2$ are information orthogonal if and only if $\supp(P_1) \cap \supp(P_2) = \emptyset$.
\end{claim}
\begin{proofsketch}
We show that if $\supp(P_1) \cap \supp(P_2) = \emptyset$, $H_1$ and $H_2$ are information orthogonal. 
W.l.o.g., we assume the parameters are indexed such that the covariance matrix $\bm{\Sigma}$ is block diagonal, with parameters from $H_1$ in one block, and parameters from $H_2$ in another block.\footnote{Strictly speaking, we need to specify the physical implementations of the probes to write the covariance matrix $\bm{\Sigma}$. Though here it is enough to think of $\bm{\Sigma}$ abstractly, as a block diagonal matrix.}
This is possible due to the fact that distinct probes are not entangled by assumption, so there is no correlation between $P_1$ and $P_2$.
Then $\bm{\Sigma}^{-1}$ is also block diagonal, with parameters from $H_1$ in one block, and parameters from $H_2$ in another block.
For any $\eta_i \in H_1$ and any $\eta_j \in H_2$, one can check that similar to~\eqref{eqn:I_{e,d}_{i,j}_term1}, the first term in $\mathcal{I}_{i,j}$~\eqref{eqn:FIM} is a rescaling of $\bm{\Sigma}^{-1}_{i,j}$.
Since $\eta_i$ and $\eta_j$ are located in separate blocks, $\bm{\Sigma}^{-1}_{i,j} = 0$.
For the second term in $\mathcal{I}_{i,j}$, it is not hard to check that the matrix multiplication of the chain of block diagonal matrices yields a zero matrix.

For the other direction, a similar reasoning gives that if $\supp(P_1) \cap \supp(P_2) \neq \emptyset$, $H_1$ and $H_2$ are no longer information orthogonal.
\end{proofsketch}

Given a set of probes $\mathcal{P}$ on a graph $G$, we can group subsets of $\mathcal{P}$ into subgraphs of $G$ by taking all vertices and edges of the probes.
Denote $V(P)$ as the set of vertices a probe $P$ passes through, including its two endpoints. 
Grouping probes $P_1, P_2, \dots, P_k$ yields a subgraph $G_1 = (V(G_1), E(G_1))$, where $V(G_1) = \bigcup_{i=1}^k V(P_i)$, and $E(G_1) = \bigcup_{i=1}^k \supp(P_i)$.
A direct consequence of Claim~\ref{clm:disjoint_probes} is that, if two subgraphs are edge-disjoint, their corresponding sets of link transmissivities are information orthogonal.
\begin{corollary}
For subgraphs $G_1 = (V(G_1), E(G_1))$ and $G_2 = (V(G_2), E(G_2))$, let $H_i$ be the set of link transmissivities associated with edges in $G_i$, $H_i = \{ \eta_i \mid \forall e_i \in E(G_i)\}$, $i=1,2$.
$H_1$ and $H_2$ are information orthogonal if and only if $E(G_1) \cap E(G_2) = \emptyset$.
\end{corollary}

Then, to maximize the number of information orthogonal sets of parameters, we just need to maximize the number of edge-disjoint subgraphs in $G$.
Note that these are not meant to be arbitrary subgraphs of $G$.
Each subgraph should be obtained by grouping some subset of the probes.
This gives rise to the definition of the subgraph cover.
\begin{definition}[Subgraph cover] \label{def:subgraph_cover}
Given a graph $G = (V(G),E(G))$ with monitors $M \subseteq V$, a set of connected subgraphs $\{G_i = (V(G_i), E(G_i))\}_{i=1}^m$ is a \emph{subgraph cover} of $G$ if
\begin{enumerate}[1)]
\item \label{cdt:subgraph_cover:nonempty}
$E(G_i) \neq \emptyset$ and $V(G_i) \cap M \neq \emptyset$, for all $i \in [m]$ (every subgraph has at least an edge and a monitor);
\item \label{cdt:subgraph_cover:edge-disj}
$E(G_i) \cap E(G_j) = \emptyset$ for all $i \neq j \in [m]$ (subgraphs are edge-disjoint); and 
\item \label{cdt:subgraph_cover:cover}
$\bigcup_{i=1}^{m} E(G_i) = E(G)$ (all edges in $G$ are covered). 
\end{enumerate}
$\{G_i = (V(G_i), E(G_i))\}_{i=1}^m$ is a \emph{maximum subgraph cover} of $G$ if $m$ is maximized.
\end{definition}

To upper bound $m$, let $\deg(S)$ be the number of edges with exactly one endpoint in $S$, $\deg(S) = \sizeof{\{(u,v) \in E \mid u \in S, v \notin S\}}$.
Denote $E(S)$ as the set of edges with both endpoints in $S$, $E(S) = \{(u,v) \in E \mid u,v \in S\}$.
\begin{fact} \label{fact:UB_m}
For any subgraph cover $\{G_i = (V(G_i), E(G_i))\}_{i=1}^m$ of a graph $G$ with monitors $M$, $m \leq \deg(M) + \sizeof{E(M)}$.
\end{fact}
\begin{proof}
Consider all edges in $E(G)$ that are incident to at least one vertex in $M$, there are $\deg(M) + \sizeof{E(M)}$ number of such edges.
Due to Condition~\ref{cdt:subgraph_cover:nonempty}, each such edge contributes at most one subgraph, while edges not incident to $M$ cannot create any new subgraph. 
It follows that $m \leq \deg(M) + \sizeof{E(M)}$.
\end{proof}
If a set of probes \emph{can} be grouped into a maximum subgraph cover, it effectively maximizes the number of information orthogonal subsets of link transmissivities, thus addressing Problem~\ref{prob}(ii).
The subtlety lies in whether the upper bound on $m$ (Fact~\ref{fact:UB_m}) is tight and, if so, whether there exists a set of probes that attains it.
The algorithm we introduce next in \S~\ref{sec:graph:algo} answers both questions in the affirmative.

\subsection{The probe construction algorithm} \label{sec:graph:algo}

\paragraph{Algorithm.}
Given a graph $G$ and a set of monitors $M$, we iteratively construct one probe for each edge in $E(G)$.
While we are not explicitly trying to minimize the length of the probes in this work, restricting our attention to the shortest \emph{loop-back} probes in most iterations is helpful in two respects.
As noted in~\cite{zheng2025quantum}, quantum probes are inherently sensitive to the distance they travel, it is therefore always desirable to use shorter probes whenever possible.
Moreover, a shortest loop-back probe passes through the least number of distinct edges possible (when $G$ is unweighted). 
A set of shorter probes are less likely to intersect on many edges, which results in better chances of maximizing information orthogonality.

Algorithm~\ref{algo:FindProbe} describes $\textsc{FindProbe}(G, M)$, our probe construction algorithm. 
For each edge incident to at least one monitor, we simply construct a probe that only passes through that edge (Steps~\ref{algo:FindProbe:step:both_incident_s}-\ref{algo:FindProbe:step:one_incident_e}).
For each edge that is further away from the monitors, similar in spirit to~\cite[Algorithm 3]{zheng2025quantum}, we use all-pair shortest paths to generate the shortest loop-back probe that passes through the edge (Steps~\ref{algo:FindProbe:step:nonincident_s}-\ref{algo:FindProbe:step:nonincident_e}).
To obtain the all-pair shortest paths (Step~\ref{algo:FindProbe:step:FW}), we opt for a version of the Floyd-Warshall algorithm (e.g., \S~25.2 in~\cite{cormen2022introduction}), which not only returns the all-pair shortest distances, but also allows the construction of the shortest paths using a predecessor array.
The notation $u_1 \rightsquigarrow u_2 \rightsquigarrow \cdots u_i \rightsquigarrow u_{i+1} \rightsquigarrow \cdots \rightsquigarrow u_k$ refers to the concatenation of the shortest paths from $u_i$ to $u_{i+1}$, $i\in [k-1]$.

\begin{algorithm}
\caption{$\textsc{FindProbe}(G, M)$}
\label{algo:FindProbe}
\textbf{Input}: A graph $G=(V(G),E(G))$, and a set of monitors $M$; \\
\textbf{Output}: A set of distinct probes $\mathcal{P}$.
\begin{algorithmic}[1]
\State Run Floyd-Warshall with path construction on $G = (V, E)$, obtain distance matrix $D$ and predecessor array $T$
\Comment{$D(u,v)$ is the shortest distance from $u$ to $v$, and $T(u,v)$ gives the penultimate vertex on the shortest path from $u$ to $v$, both in $G$.}
\label{algo:FindProbe:step:FW}
\State $\mathcal{P} \gets \emptyset$
\For{$(u,v) \in E$}
    \If{both $u,v \in M$} \label{algo:FindProbe:step:both_incident_s}
        \State Construct $P$ to be $u \rightsquigarrow v$ \label{algo:FindProbe:step:both_incident_e}
    \ElsIf{Exactly one of $u,v \in M$ (w.l.o.g. say $u \in M$)} \label{algo:FindProbe:step:one_incident_s}
        \State Construct $P$ to be $u \rightsquigarrow v \rightsquigarrow u$ \label{algo:FindProbe:step:one_incident_e}
    \Else
        \State $m_u \gets \argmin_{m \in M} D(m,u)$ \label{algo:FindProbe:step:nonincident_s}
        \State $m_v \gets \argmin_{m \in M} D(v,m)$ \label{algo:FindProbe:step:nonincident_mv}
        \State Let $P$ be the shorter of $m_u \rightsquigarrow u \rightsquigarrow v \rightsquigarrow u \rightsquigarrow m_u$ and $m_v \rightsquigarrow v \rightsquigarrow u \rightsquigarrow v \rightsquigarrow m_v$ (If $D(m_u, u) = D(m_v, v)$, let $P$ be $m_u \rightsquigarrow u \rightsquigarrow v \rightsquigarrow u \rightsquigarrow m_u$ if and only if $m_u$ is lexicographically smaller than $m_v$) \label{algo:FindProbe:step:nonincident_e}
    \EndIf
    \State Add $P$ to $\mathcal{P}$
\EndFor
\State \Return $\mathcal{P}$
\end{algorithmic}
\end{algorithm}

\paragraph{Correctness.}
Next we show that Algorithm~\ref{algo:FindProbe} correctly outputs a solution to Problem~\ref{prob}.
We start by showing that the probes generated by Algorithm~\ref{algo:FindProbe} guarantee identifiability.
\begin{lemma}
Given a graph $G$ with monitors $M$ and unknown link transmissivities $\bm{\eta} = (\eta_1, \eta_2, \dots, \eta_n)$, $\bm{\eta}$ is identifiable from $\mathcal{P} = \textsc{FindProbe}(G, M)$.
\end{lemma}
\begin{proofsketch}
We sketch a proof using induction on the length of the probes generated by $\textsc{FindProbe}(G, M)$.
For the base case, we have probes of length $1$ or $2$ (from Steps~\ref{algo:FindProbe:step:both_incident_s}-\ref{algo:FindProbe:step:one_incident_e}).
These probes pass through exactly $1$ edge incident to $M$.
The transmissivity of each such edge $(u_1,v_1)$ is identifiable from the probe that only passes through $(u_1,v_1)$.
Assuming all edges on probes of length at most $2k$ are identifiable for some positive integer $k$, we have that on a probe $P$ with length $2(k+1)$, the transmissivities of all but one edge $(u_2, v_2)$ are identifiable by the induction hypothesis, with $(u_2, v_2)$ being the edge furthest away from $M$.
Then the observation from $P$, together with the other transmissivities that are identifiable on $P$, suffice to identify the transmissivity of $(u_2, v_2)$.

In fact, the induction argument mirrors the steps of Gaussian elimination applied to the measurement matrix $A$ associated with $\mathcal{P}$.
At the end of the induction, up to relabeling of the probes, $A$ is in row echelon form, with exactly $n$ nonzero pivots.
We can then conclude that $A$ has full column rank, thus making $\bm{\eta}$ identifiable from $\mathcal{P}$.
\end{proofsketch}

To demonstrate $\mathcal{P}=\textsc{FindProbe}(G, M)$ also maximizes the number of information orthogonal subsets of link transmissivities, we construct a sequence of subgraphs $\{G_i\}_{i=1}^{m^*}$ from $\mathcal{P}$, check that $\{G_i\}_{i=1}^{m^*}$ is a subgraph cover of $G$, and prove that $\{G_i\}_{i=1}^{m^*}$ is maximum by attaining the upper bound in Fact~\ref{fact:UB_m}.
Since $\mathcal{P}$ can be grouped into a maximum subgraph cover, it follows that $\mathcal{P}$ satisfies satisfies Problem~\ref{prob}(ii).

\begin{lemma}
Given a graph $G$ with monitors $M$ and unknown link transmissivities $\{\eta_1, \eta_2, \dots, \eta_n\}$, $\mathcal{P} = \textsc{FindProbe}(G, M)$ maximizes the number of information orthogonal subsets of $\{\eta_1, \eta_2, \dots, \eta_n\}$.
\end{lemma}
\begin{proof}
We start by constructing a sequence of subgraphs $\{G_i\}_{i=1}^{m^*}$ from $\mathcal{P}$.
Let $\mathcal{P}_1, \mathcal{P}_2$ and $\mathcal{P}_3$ be the sets of probes generated in Steps~\ref{algo:FindProbe:step:both_incident_e}, \ref{algo:FindProbe:step:one_incident_e} and~\ref{algo:FindProbe:step:nonincident_e} of Algorithm~\ref{algo:FindProbe}, respectively.
Clearly, $\mathcal{P}_1 \cup \mathcal{P}_2 \cup \mathcal{P}_3 = \mathcal{P}$.
\begin{enumerate}[(a)]
\item \label{constr:initial}
For each $P_i \in \mathcal{P}_1 \cup \mathcal{P}_2$, $\supp(P_i)$ is an edge incident to some vertex in $M$.
Define subgraph $G_i$ for each $P_i$, where $V(G_i) = V(P_i)$, and $E(G_i) = \supp(P_i)$.
\item \label{constr:update}
For each $P_j \in \mathcal{P}_3$, by Step~\ref{algo:FindProbe:step:nonincident_e}, there exists a unique $P_{i^*} \in \mathcal{P}_2$, such that $\supp(P_{i^*}) \subseteq \supp(P_j)$.
Then add $P_j$ to $G_{i^*}$, that is, update $V(G_{i^*}) \gets V(G_{i^*}) \cup V(P_j)$, and $E(G_{i^*}) \gets E(G_{i^*}) \cup \supp(P_j)$.
\end{enumerate}

Next we verify that $\{G_i\}_{i=1}^{m^*}$ is indeed a subgraph cover.
Condition~\ref{cdt:subgraph_cover:nonempty} in Definition~\ref{def:subgraph_cover} is guaranteed to hold by~\ref{constr:initial}.
Condition~\ref{cdt:subgraph_cover:cover} is satisfied by the fact that Algorithm~\ref{algo:FindProbe} iteratively generates one probe $P$ for each edge $(u,v) \in E(G)$, and $(u,v) \in \supp(P)$.
To prove that Condition~\ref{cdt:subgraph_cover:edge-disj} is also satisfied, suppose for contradiction that there exists subgraphs $G_1$ and $G_2$ in $\{G_i\}_{i=1}^{m^*}$, such that $E(G_1) \cap E(G_2) \neq \emptyset$.
There must exist $P_1$ from the construction of $G_1$, and $P_2$ from the construction of $G_2$, such that some edge $(u,v) \in \supp(P_1) \cap \supp(P_2)$.
Let $m_u$ and $m_v$ be the monitors on $P_1$ and $P_2$, respectively.
By the construction of subgraphs, we have that $m_u \neq m_v$.
W.l.o.g., assume $m_u$ is the monitor closer to $u$, and $m_v$ the monitor closer to $v$.
Then $P_1$ corresponds to $m_u \rightsquigarrow u \rightsquigarrow v \rightsquigarrow \cdots \rightsquigarrow v \rightsquigarrow u \rightsquigarrow m_u$, and $P_2$ corresponds to $m_v \rightsquigarrow v \rightsquigarrow u \rightsquigarrow \cdots \rightsquigarrow u \rightsquigarrow v \rightsquigarrow m_v$.
In constructing $P_1$, Steps~\ref{algo:FindProbe:step:nonincident_s}-\ref{algo:FindProbe:step:nonincident_e} ensures that $D(m_u, u) \leq D(m_v, v)$.
Similarly from $P_2$, we have $D(m_u, u) \geq D(m_v, v)$.
If $D(m_u, u) \neq D(m_v, v)$, we have a contradiction.
When $D(m_u, u) = D(m_v, v)$, Step~\ref{algo:FindProbe:step:nonincident_e} indicates that we always breaks ties consistently.
That is, either $P_1$ and $P_2$ both pass through $m_u \rightsquigarrow u$, or they both pass through $m_v \rightsquigarrow v$.
This contradicts the fact that $m_u \neq m_v$.

Finally, observe that only~\ref{constr:initial} created $\sizeof{\mathcal{P}_1 \cup \mathcal{P}_2}$ subgraphs, while~\ref{constr:update} updated the existing subgraphs.
Therefore, we have that $\{G_i\}_{i=1}^{m^*}$ is a maximum subgraph cover with $m^* = \deg(M) + \sizeof{E(M)}$.
\end{proof}

\paragraph{Time complexity.}
Suppose the given graph $G$ has $\sizeof{V}$ vertices, $\sizeof{E}$ edges, and $\sizeof{M}$ monitors. 
Running Floyd-Warshall on $G$ (Step~\ref{algo:FindProbe:step:FW}) takes time $O(\sizeof{V}^3)$.
For each edge incident to $M$, it only takes constant time to construct a probe (Steps~\ref{algo:FindProbe:step:both_incident_s}-\ref{algo:FindProbe:step:one_incident_e}).
These edges together contribute $O(\deg(M) + \sizeof{E(M)})$ time.
For each of the rest of the $\sizeof{E} - (\deg(M) + \sizeof{E(M)})$ edges that is further away from $M$, locating monitors closest to each endpoint (Steps~\ref{algo:FindProbe:step:nonincident_s}-\ref{algo:FindProbe:step:nonincident_mv}) takes time $O(\sizeof{M})$.
Constructing a probe (Step~\ref{algo:FindProbe:step:nonincident_e}) requires at most $O(\sizeof{V})$ time.
The time it takes to report $n$ probes is clearly subsumed by previous steps of the algorithm.
The overall time complexity of $\textsc{FindProbe}(G, M)$ adds up to $O(\sizeof{V}^3 + \sizeof{V} \cdot \sizeof{E} - \sizeof{V} \cdot (\deg(M) + \sizeof{E(M)}))$.
Even though the Floyd-Warshall algorithm dominates the runtime, it is clear that the more monitors we have, the more shorter probes we can construct, and the faster is Algorithm~\ref{algo:FindProbe}.

%% file: general.tex
\section{Quantum improvement in a network}
\label{sec:combined}

Given two sets of probes $\mathscr{P}$ and $\mathscr{P}'$, each capable of identifying all links in a network, which set better estimates the link transmissivities?
To address this question, we consider the FIM induced by a set of probes in the general network case, and characterize the determinant of the FIM, and the trace of its inverse.
Recall that each probe $(P, \verb|impl|(P), t(P), c(P)) \in \mathscr{P}$ corresponds to a channel with FI $I_P$.
In \S~\ref{sec:combined:lemma}, we shall see that both the determinant of the FIM and the trace of its inverse relate neatly to the channel FI of the individual probes.
Finally, we show why the two metrics admit particularly clean closed-forms in \S~\ref{sec:combined:proof}.

\subsection{The two metrics for a general network} \label{sec:combined:lemma}

Our main result below provides explicit expressions for the the determinant of the FIM, and the trace of its inverse:

\begin{lemma} \label{lem:det_FIM_graph}
Given a set of probes $\mathscr{P} = \{(P_i, \verb|impl|(P_i), t(P_i), c(P_i))\}_{i=1}^n$ that identifies all link transmissivities $\{\eta_1, \eta_2, \dots, \eta_n\}$ of a graph $G$, $\mathscr{P}$ induces a measurement matrix $A$ and a FIM $\mathcal{I}$.
Then 
\begin{align}
& \det(\mathcal{I}) 
= \left( \prod_{i=1}^n \eta_i^{-2} \right) \cdot \det(A)^2 \cdot \left( \prod_{i=1}^n c(P_i) \cdot \eta_{P_i}^2 I_{P_i} \right), \label{eqn:det_FIM} \\
& \Tr(\mathcal{I}^{-1})
= \sum_{i=1}^n \eta_i^2 \sum_{j=1}^n \frac{(A^{-1})_{i,j}^2}{c(P_j)\cdot \eta_{P_j}^2 I_{P_j}}, \label{eqn:tr_inverse}
\end{align}
where $\eta_{P_i}$ and $I_{P_i}$ are the channel transmissivity and the FI corresponding to the probe $(P_i, \verb|impl|(P_i), t(P_i), c(P_i))$.
\end{lemma}

First we clarify that, since $n$ link transmissivities are identifiable by $n$ $4$-tuple probes, $P_1, P_2, \dots, P_n$ must be distinct.
Let $\mathcal{P} = \{P_1, P_2, \dots, P_n\}$ be the set of graph-theoretic probes.
We have that the measurement matrix $A$, defined solely based on $\mathcal{P}$, must be a non-singular matrix, $\det(A)$ is therefore well-defined.

A notable feature of Lemma~\ref{lem:det_FIM_graph} is that it relates the two metrics of interest in the general setting to the FI of the channel case, while decoupling the physical implementation of the probes from their graph-theoretic properties.
This feature makes it especially convenient when comparing different physical implementations given the same set of graph-theoretic probes.
For example, given two sets of probes $\mathscr{P}$ and $\mathscr{P}'$, both with the same set of graph-theoretic probes $\mathcal{P} = \{P_1, P_2, \dots, P_n\}$, $\mathscr{P}$ and $\mathscr{P}'$ only differ in the $\verb|impl|(P_i), t(P_i)$ and $c(P_i)$ on some subset of the $P_i$.
Since the dependence on $t(P_i)$ is implicit in the channel FI $I_{P_i}$, when switching from $\mathscr{P}$ to $\mathscr{P}'$, only $c(P_i)$ and $I_{P_i}$ may change in~\eqref{eqn:det_FIM} and~\eqref{eqn:tr_inverse}.
Using the prime $(')$ to label the quantities induced by $\mathscr{P}'$, we have that the ratio $\frac{\det(\mathcal{I'})}{\det(\mathcal{I'})}$ completely depends on the physical implementations of the probes, while $\Tr(\mathcal{I'}^{-1}) - \Tr(\mathcal{I}^{-1})$ is a weighted sum of $\frac{1}{c(P_i)' I'_{P_i}} - \frac{1}{c(P_i) I_{P_i}}$, with weights solely determined by the graph-theoretic properties of $\mathcal{P}$:
\begin{align}
\frac{\det(\mathcal{I'})}{\det(\mathcal{I'})}
& = \prod_{i=1}^n \frac{c(P_i)' \cdot I'_{P_i}}{c(P_i) \cdot I_{P_i}}, \label{eqn:det_improvement} \\
\Tr(\mathcal{I'}^{-1}) - \Tr(\mathcal{I}^{-1})
& = \sum_{i=1}^n \sum_{j=1}^n \frac{\eta_i^2 (A^{-1})_{i,j}^2}{\eta_{P_j}^2}\left( \frac{1}{c(P_i)' I'_{P_i}} - \frac{1}{c(P_i) I_{P_i}} \right) \label{eqn:tr_improvement}.
\end{align}
In particular, if $\mathscr{P}$ is implemented by classical probes, and $\mathscr{P}'$ is quantum,~\eqref{eqn:det_improvement} and~\eqref{eqn:tr_improvement} together with Fact~\ref{fact:single_ch_q_vs_c} characterize the quantum improvement in a general network.

In fact, Lemma~\ref{lem:det_FIM_graph} could also apply in other contexts.
While we have not focused on finding the optimal probes under resource constraints,~\eqref{eqn:det_FIM} and~\eqref{eqn:tr_inverse} could serve as objective functions should such optimization problem arises.
When only local changes are made to a small subset of the probes, many factors in these expressions remain unchanged, making them efficient to evaluate.
A subtle issue, however, is that evaluating~\eqref{eqn:det_FIM} and~\eqref{eqn:tr_inverse} requires prior knowledge of $\eta_1, \eta_2, \dots, \eta_n$, the very parameters we seek to estimate in teh tomography problem.
A potential solution is to first obtain rough estimates of these link transmissivities, and use the estimates in the objective function.
We leave this optimization problem for future work.

\subsection{Proof of Lemma~\ref{lem:det_FIM_graph}} \label{sec:combined:proof}
To prove Lemma~\ref{lem:det_FIM_graph}, it is important to understand the structure of the covariance matrix $\bm{\Sigma}$ induced by $\mathscr{P}$.
Notice that for any pair of distinct $P_1$ and $P_2$ from $\mathcal{P}$, regardless of whether $P_1$ and $P_2$ share edges, observations from $(P_1, \verb|impl|(P_1), t(P_1), m(P_1))$ and $(P_2, \verb|impl|(P_2), t(P_2), m(P_2))$ are always uncorrrelated.
This implies that $\bm{\Sigma}$ is block diagonal, where a probe $(P, \verb|impl|(P), t(P), m(P))$ occupies $m(P)$ number of blocks of size $t(P) \times t(P)$ each (a fact glossed over in the proof sketch of Claim~\ref{clm:disjoint_probes}).
In particular, a probe $(P, \verb|impl|(P), t(P), m(P))$ implemented by coherent or squeezed states, i.e., $\verb|impl|(P) = `\verb|coherent|\textrm'$ or $`\verb|squeezed|\textrm'$, contributes $m(P)$ elements on the the diagonal of $\bm{\Sigma}$.
It is also implicit from \S~\ref{sec:physical:single} that $\bm{\Sigma}^{-1}$ preserves the same block diagonal structure.
Denote $\mathcal{I}_{i,j}^{(P)}$ as the contribution to $\mathcal{I}_{i,j}$ from probe $(P, \verb|impl|(P), t(P), m(P))$.
Since $\bm{\Sigma}$ and $\bm{\Sigma}^{-1}$ are both block diagonal, \eqref{eqn:FIM} implies
\begin{equation} \label{eqn:I_ij_sum_contrib}
\mathcal{I}_{i,j}
= \sum_{P \in \mathcal{P}} m(P) \cdot \mathcal{I}_{i,j}^{(P)}.
\end{equation}
A direct application of~\eqref{eqn:FIM} to the channel corresponding to $(P, \verb|impl|(P), t(P), m(P))$ also gives the channel FI
\begin{equation} \label{eqn:I_P}
I_P
= \frac{\partial \bm{\mu}_P^T}{\partial \eta_P} \bm{\Sigma_P}^{-1} \frac{\partial \bm{\mu}_P}{\partial \eta_P}  + \frac{1}{2}\Tr \left(\bm{\Sigma_P}^{-1} \frac{\partial \bm{\Sigma}_P}{\partial \eta_P} \bm{\Sigma}_P^{-1}\frac{\partial \bm{\Sigma_P}}{\partial \eta_P}\right).
\end{equation}
Following~\eqref{eqn:I_P} and the chain rule,
\begin{align}
\mathcal{I}_{i,j}^{(P)}
& = \frac{\partial \bm{\mu}_P^T}{\partial \eta_i} \bm{\Sigma_P}^{-1} \frac{\partial \bm{\mu}_P}{\partial \eta_j}  + \frac{1}{2}\Tr \left(\bm{\Sigma_P}^{-1} \frac{\partial \bm{\Sigma}_P}{\partial \eta_i} \bm{\Sigma}_P^{-1}\frac{\partial \bm{\Sigma_P}}{\partial \eta_j}\right) \nonumber \\
& = \frac{\partial \bm{\mu}_P^T}{\partial \eta_P} \cdot \frac{\partial \eta_P}{\partial \eta_i} \bm{\Sigma_P}^{-1} \frac{\partial \bm{\mu}_P}{\partial \eta_P} \cdot \frac{\partial \eta_P}{\partial \eta_j} 
+ \frac{1}{2}\Tr \left(
\bm{\Sigma_P}^{-1} \frac{\partial \bm{\Sigma}_P}{\partial \eta_P} \cdot  \frac{\partial \eta_P}{\partial \eta_i} 
\bm{\Sigma}_P^{-1} \frac{\partial \bm{\Sigma_P}}{\partial \eta_P} \cdot  \frac{\partial \eta_P}{\partial \eta_j} \right) \nonumber \\
& = \frac{\partial \eta_P}{\partial \eta_i} \cdot \frac{\partial \eta_P}{\partial \eta_j} \cdot I_P. \label{eqn:I_ij_P_initial}
\end{align}
Here in $\mathcal{I}_{i,j}^{(P)}$, we can already observe a separation between the graph-theoretic properties of $(P, \verb|impl|(P), t(P), m(P))$ and its physical implementation --- the factor $\frac{\partial \eta_P}{\partial \eta_i} \cdot \frac{\partial \eta_P}{\partial \eta_j}$ only depends on how the probe is routed in the network.

Next we express $\frac{\partial \eta_P}{\partial \eta_i}$ in terms of the measurement matrix $A$.
Recall that each row in the $A$ corresponds to a probe $P$, and each column in $A$ corresponds to an edge $e$.
With a slight abuse of notation, we write $A_{P,e}$ as the entry in $A$ corresponding to probe $(P, \verb|impl|(P), t(P), m(P))$ and edge $e$, and $A_{P,e}$ specifies the multiplicity of $e$ in the multiset $P$.
Then we have $\eta_P = \prod_{e \in \supp(P)} \eta_e^{A_{P,e}}$, and $\ln{\eta_P} = \sum_{e \in \supp(P)} A_{P,e} \ln{\eta_e}$.
Taking the partial derivative with respect to $\eta_i$ (the transmissivity of edge $e_i$) on both sides, we have 
$\frac{\partial \ln{\eta_P}}{\partial \eta_i} 
A_{P, e_i} \cdot \frac{\partial \ln{\eta_i}}{\partial \eta_i} 
= \frac{A_{P, e_i}}{\eta_i}$.
Applying the chain rule to the LHS leads to 
$\frac{\partial \ln{\eta_P}}{\partial \eta_P} \cdot \frac{\partial \eta_P}{\partial \eta_i} 
= \frac{1}{\eta_P} \cdot \frac{\partial \eta_P}{\partial \eta_i} 
= \frac{A_{P, e_i}}{\eta_i}$, which gives
$\frac{\partial \eta_P}{\partial \eta_i} 
= A_{P, e_i} \cdot \frac{\eta_P}{\eta_i}$.
Therefore, we can rewrite~\eqref{eqn:I_ij_P_initial} as
$\mathcal{I}_{i,j}^{(P)}
= A_{P, e_i} A_{P, e_j} \cdot \frac{\eta_P^2}{\eta_i \eta_j} I_P$.
Then by~\eqref{eqn:I_ij_sum_contrib},
\begin{equation} \label{eqn:I_ij}
\mathcal{I}_{i,j}
= \sum_{P \in \mathcal{P}} m(P) \cdot A_{P, e_i} A_{P, e_j} \cdot \frac{\eta_P^2}{\eta_i \eta_j} I_P.
\end{equation}

Let $\bm{D}_{\bm{\eta}} = \diag(\eta_1, \eta_2, \dots, \eta_n)$ and $\bm{D}_{\mathscr{P}} = \diag(c(P_1)\cdot \eta_{P_1}^2 I_{P_1}, c(P_2)\cdot \eta_{P_2}^2 I_{P_2}, \dots, c(P_n)\cdot \eta_{P_n}^2 I_{P_n})$.
One can check that~\eqref{eqn:I_ij} is in fact consistent with $\mathcal{I}$ being the following matrix product:
\begin{equation} \label{eqn:FIM_mat_prod}
\mathcal{I}
= \bm{D}_{\bm{\eta}}^{-1} A^T \bm{D}_{\mathscr{P}} A \bm{D}_{\bm{\eta}}^{-1}.
\end{equation}
Therefore, \eqref{eqn:det_FIM} follows from the fact that $A$, $\bm{D}_{\bm{\eta}}$ and $\bm{D}_{\mathscr{P}}$ are all non-singular matrices.

For \eqref{eqn:tr_inverse}, by~\eqref{eqn:FIM_mat_prod} and $(A^T)^{-1} = (A^{-1})^T$,
\begin{equation*}
\mathcal{I}^{-1} 
= \bm{D}_{\bm{\eta}} A^{-1} \bm{D}_{\mathscr{P}}^{-1} (A^{-1})^T \bm{D}_{\bm{\eta}}. 
\end{equation*}
Using the cyclic property of the trace,
\begin{align*}
\Tr(\mathcal{I}^{-1})
& = \Tr(\bm{D}_{\bm{\eta}}^2 A^{-1} \bm{D}_{\mathscr{P}}^{-1} (A^{-1})^T) \\
& = \sum_{i=1}^n \eta_i^2 (A^{-1} \bm{D}_{\mathscr{P}}^{-1} (A^{-1})^T)_{i,i} \\
& = \sum_{i=1}^n \eta_i^2 \sum_{j=1}^n (A^{-1})_{i,j} \cdot \frac{1}{c(P_j)\cdot \eta_{P_j}^2 I_{P_j}} \cdot (A^{-1})^T_{i,j} \\
& = \sum_{i=1}^n \eta_i^2 \sum_{j=1}^n \frac{(A^{-1})_{i,j}^2}{c(P_j)\cdot \eta_{P_j}^2 I_{P_j}}.
\end{align*}


%% file: related.tex
\section{Related Work} \label{sec:related}

\paragraph{The probes.}
The quantum and classical probes in this work have previously been used in the context of detecting drops in link transmissivity within a channel~\cite{guha2025quantum}, and localizing transmission loss change in a general network~\cite{zheng2025quantum}.
In both cases, the quantum speedup essentially comes from the fact that quantum probes induce a larger Kullback–Leibler divergence from the pre-change to the post-change distributions.
Here we use these probes to estimate link transmissivities, and the quantum improvement is quantified by the determinannt of the FIM, and the trace of its inverse.

\paragraph{Sharing entanglement.}
When we go beyond a single channel and look at a general network, an interesting question that arises is whether sharing entanglement across probes traversing different paths is beneficial.
While we do not yet have a rigorous argument, the evidence in \S~\ref{sec:physical:multi} and Appendix~\ref{apd:2_channel} points to a negative answer.
This is in fact unsurprising.
In multiparameter quantum metrology, it has been shown that entanglement does not provide a systematic advantage over displaced squeezed states~\cite{gessner2020multiparameter,nichols2018multiparameter}.
It remains open whether we can leverage techniques from quantum metrology to formally prove such no-go results in network tomography.

\paragraph{Probe construction algorithms.}
A large body of work has studied how to route the probes in a network~\cite{zheng2025quantum,ma2013efficient,ahuja2011srlg,zhao2009towards}, however, to the best of our knowledge, none has focused on maximizing the number of information orthogonal sets of unknown parameters.
Algorithm~\ref{algo:FindProbe} bears the closest resemblance to the probe construction algorithm in~\cite{zheng2025quantum}.
While we have not explicitly tried to minimize the length of the longest probe (a central goal in~\cite{zheng2025quantum}), it is not hard to see that Algorithm~\ref{algo:FindProbe} achieves this goal approximately: the longest probe it generates exceeds the minimum possible by at most $1$ if the underlying graph is unweighted.

\paragraph{The metrics.}
While we think of each link in the network as a pure-loss bosonic channel, prior works in quantum network tomography have modeled the links as bit-flip~\cite{de2024quantum}, depolarizing~\cite{de2024quantum} and Pauli channels~\cite{wang2025quantum}.
In both~\cite{wang2025quantum} and~\cite{de2024quantum}, the trace of the inverse quantum FIM has been used to evaluate the performance of the estimation.
The two metrics in this work, the determinant of the FIM, and the trace of its inverse, have been studied in finding the optimal probe allocation in classical network tomography~\cite{he2015fisher}.
Although the packet loss tomography in~\cite{he2015fisher} estimates upper-layer link success rates, whereas we focus on physical-layer link transmissivity, and the observations from the probes follow entirely different distributions, we nevertheless observe a striking similarity between our FIM~\eqref{eqn:FIM_mat_prod} and (17) in~\cite{he2015fisher}.

%% file: conclusion.tex
\section{Conclusion} \label{sec:conclusion}

In this paper, we consider the problem of estimating link transmissivities in optical networks using quantum probes.
To decide how to route the probes in a network, we propose a probe construction algorithm that not only guarantees identifiability, but also maximizes the number of information orthogonal sets of link transmissivities.
To measure the performance of a given set of probes, we derive closed-form expressions for the determinant of the FIM, and the trace of its inverse.
In particular, both metrics can be used to quantify the quantum improvement in this estimation problem.
As noted at the end of \S~\ref{sec:combined:lemma}, both metrics could potentially serve as objective functions in the associated optimization problem, namely, finding the optimal set of probes under resource constraints. 
We leave this for future work.

\section*{Acknowledgment}

The initial stage of this research was supported by the DARPA Quantum Augmented Networking (QuANET) Program under contract number HR001124C0405.

%% file: appendix.tex
\section{Sharing entanglement across two probes that share links} \label{apd:2_channel}

\begin{figure}
\centering
\begin{subfigure}[b]{0.48\textwidth}
    \centering
    \includegraphics[width=\textwidth]{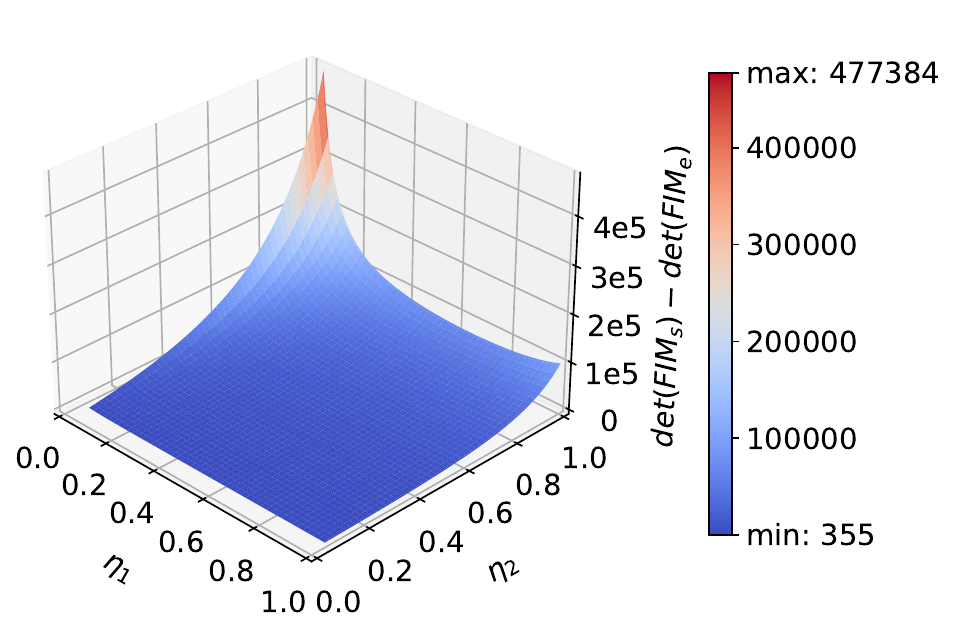}
    \caption{$\det(\mathcal{I}_{\mathrm{s,s}}) - \det(\mathcal{I}_{\mathrm{e,s}})$ wrt $\eta_1$ and $\eta_2$.}
    \label{fig:2_channel:det}
\end{subfigure}
\hfill
\begin{subfigure}[b]{0.48\textwidth}
    \centering
    \includegraphics[width=\textwidth]{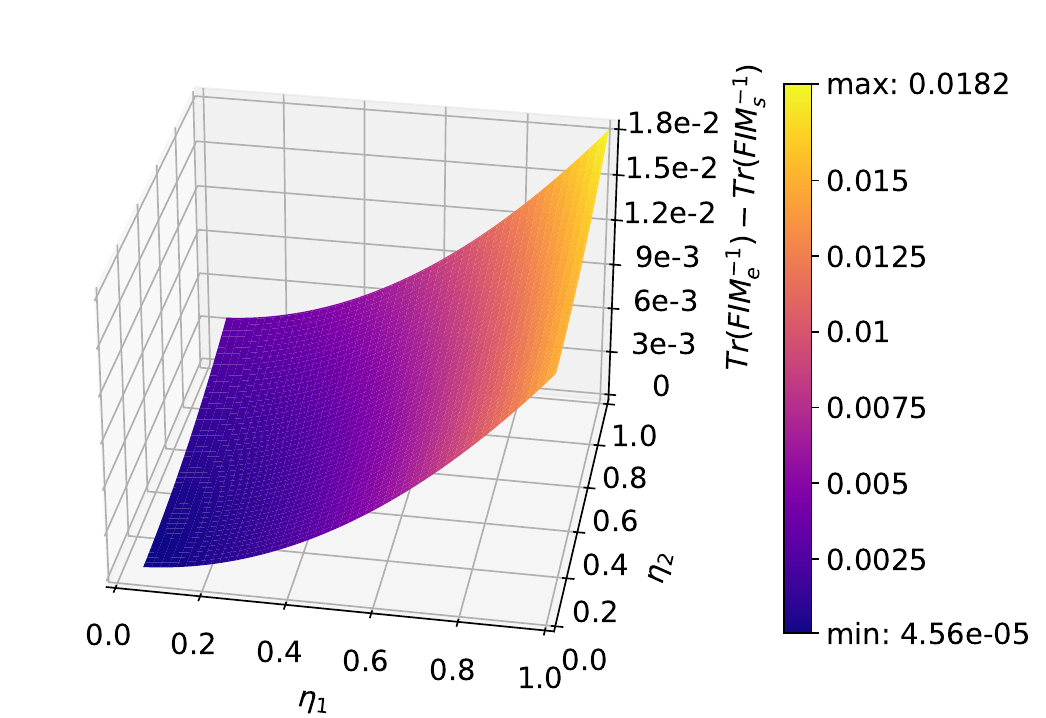}
    \caption{$\Tr(\mathcal{I}_{\mathrm{e,s}}^{-1}) - \Tr(\mathcal{I}_{\mathrm{s,s}}^{-1})$ wrt $\eta_1$ and $\eta_2$.}
    \label{fig:2_channel:Tr}
\end{subfigure}
\caption{Differences in the two metrics for the two-channel setup.}
\end{figure}

Consider two probes that are set through two channels with transmissivities $\eta_1\eta_2$ and $\eta_2$, respectively.
Fix the classical and quantum energy per-pulse, $N$ and $N_a$, we compare two cases:
\begin{enumerate*}[(i)]
\item each probe is a displaced squeezed state; and 
\item the two probes are entangled.
\end{enumerate*}
We present numerical results demonstrating that, similar to the independent channel configuration in \S~\ref{sec:physical:multi}, with respect to \emph{both} the determinant of the FIM and the trace of its inverse, entangled probes perform worse than the ones implemented by independent squeezing.

We use the subscript `$\cdot\mathrm{,s}$' to distinguish this setup, where the channel transmissivities have shared factors.
Recall that $c_n$~\eqref{eqn:c_n} hides the quantum energy $N_a$.
When the two probes are implemented by squeezed states, we have
\begin{align*}
& \bm{\mu}_{\mathrm{s,s}} 
= \sqrt{N} \cdot 
\begin{pmatrix}
\sqrt{\eta_1 \eta_2} \\
\sqrt{\eta_2}
\end{pmatrix}, \\
& \bm{\Sigma}_{\mathrm{s,s}} 
= 
\begin{pmatrix}
\frac{1}{4} - \frac{c_1 \eta_1 \eta_2}{4} & 0 \\
0 & \frac{1}{4} - \frac{c_1\eta_2}{4}
\end{pmatrix},
\end{align*}
\begin{align}
\det(\mathcal{I}_{\mathrm{s,s}}) 
& = \frac{1}{4\eta_1} 
\cdot \frac{2N(1-c_1\eta_1 \eta_2)+c_1^2\eta_1 \eta_2}{(1- c_1 \eta_1 \eta_2)^2} 
\cdot \frac{2N(1-c_1\eta_2)+c_1^2\eta_2}{(1- c_1 \eta_2)^2},  \nonumber \\
\Tr(\mathcal{I}_{\mathrm{s,s}}^{-1})
& = \frac{2}{\eta_2} \cdot
\left( 
\frac{\eta_1(1- c_1 \eta_1 \eta_2)^2}{2N(1-c_1\eta_1 \eta_2)+c_1^2\eta_1 \eta_2}
+ \frac{(\eta_1^2 + \eta_2^2)(1- c_1 \eta_2)^2}{2N(1-c_1\eta_2)+c_1^2\eta_2}
\right). \nonumber
\end{align}
If the two probes are entangled,
\begin{align*}
& \bm{\mu}_{\mathrm{e,s}} 
= \sqrt{N} \cdot 
\begin{pmatrix}
\sqrt{\eta_1 \eta_2} \\
\sqrt{\eta_2}
\end{pmatrix}, \\
& \bm{\Sigma}_{\mathrm{s,s}} 
= 
\begin{pmatrix}
\frac{1}{4} - \frac{c_2 \eta_1 \eta_2}{8} & -\frac{c_2\sqrt{\eta_1} \eta_2}{8} \\
-\frac{c_2\sqrt{\eta_1} \eta_2}{8} & \frac{1}{4} - \frac{c_2\eta_2}{8}
\end{pmatrix},
\end{align*}
\begin{align}
\det(\mathcal{I}_{\mathrm{e,s}}) 
& = \frac{32N^2 (2-c_2\eta_2(1+\eta_1))^2 + 12Nc_2^2 \eta_2(1+\eta_1)(2-c_2\eta_2(1+\eta_1)) + c_2^4\eta_2^2(1+\eta_1)^2}{16\eta_1(2-c_2\eta_2(1+\eta_1))^3},  \nonumber \\
\Tr(\mathcal{I}_{\mathrm{s,s}}^{-1})
& = \frac{2(2-c_2 \eta_2(1+\eta_1))}{(1+\eta_1)\eta_2} \cdot
\left( 
\frac{4\eta_1 ((1+\eta_1)^2 + \eta_2^2)}{16N-c_2\eta_2(8N-c_2)(1+\eta_1)}
+ \frac{\eta_2^2(2-c_2\eta_2(1+\eta_1))}{8N-c_2\eta_2(4N-c_2)(1+\eta_1)}
\right). \nonumber
\end{align}

We compute $\det(\mathcal{I}_{\mathrm{s,s}}) - \det(\mathcal{I}_{\mathrm{e,s}})$ (Figure~\ref{fig:2_channel:det}) and $\Tr(\mathcal{I}_{\mathrm{e,s}}^{-1}) - \Tr(\mathcal{I}_{\mathrm{s,s}}^{-1})$ (Figure~\ref{fig:2_channel:Tr}) as $\eta_1, \eta_2$ range from $0$ to $1$.
Numerical evidence suggests that, we indeed have $\det(\mathcal{I}_{\mathrm{s,s}}) > \det(\mathcal{I}_{\mathrm{e,s}})$ and $\Tr(\mathcal{I}_{\mathrm{s,s}}^{-1}) < \Tr(\mathcal{I}_{\mathrm{e,s}}^{-1})$ for all $\eta_1, \eta_2 \in (0,1]$.

\section{The structure of $\mathcal{I}_{\mathrm{e,d}}$ (proof of Claim~\ref{clm:I_{e,d}})} \label{apd:proof_I_{e,d}}

To derive a closed-form for $\mathcal{I}_{\mathrm{e,d}}$, we mechanically derive each element that appears in~\eqref{eqn:FIM} for $(\mathcal{I}_{\mathrm{e,d}})_{i,j}$.
First denote $\bm{e}_i$ as the $i$-th unit vector, then $\frac{\partial \bm{\mu}_{\mathrm{e,d}}}{\partial \eta_i} = \frac{\sqrt{N}}{2\sqrt{\eta_i}} \bm{e}_i$ by~\eqref{eqn:mu_ent_split}.
Because $\bm{\Sigma}_{\mathrm{e,d}}$~\eqref{eqn:Sigma_ent_split} is a rank-$1$ approximation of $\frac{1}{4} \bm{I}$, the Sherman-Morrison formula applies, and $\bm{\Sigma}_{\mathrm{e,d}}^{-1} = 4 \bm{I} + \frac{4c_n}{n-c_n S_{\bm{\eta}}}\bm{v} \bm{v}^T$. (Recall that $\bm{v} = (\sqrt{\eta_1}, \sqrt{\eta_2}, \dots, \sqrt{\eta_n})$.)
The first term in~\eqref{eqn:FIM} hence becomes
\begin{equation} \label{eqn:I_{e,d}_{i,j}_term1}
\frac{\partial \bm{\mu}_{\mathrm{e,d}}^T}{\partial \eta_i} \cdot \bm{\Sigma}_{\mathrm{e,d}}^{-1} \cdot \frac{\partial \bm{\mu}_{\mathrm{e,d}}}{\partial \eta_j}
= \frac{N}{4\sqrt{\eta_i \eta_j}} \cdot \left(\bm{\Sigma}_{\mathrm{e,d}}^{-1} \right)_{i,j}
= \frac{c_nN}{n-c_nS_{\bm{\eta}}} + N \cdot \frac{\mathbbm{1}[i=j]}{\sqrt{\eta_i \eta_j}} \cdot .
\end{equation}
Therefore, the contribution to $\mathcal{I}_{\mathrm{e,d}}$ from the first term in~\eqref{eqn:FIM} is 
\begin{equation*}
N \cdot \operatorname{diag}(\frac{1}{\eta_1}, \frac{1}{\eta_2}, \dots, \frac{1}{\eta_n}) + \frac{c_nN}{n-c_nS_{\bm{\eta}}} \cdot \bm{u} \bm{u}^T.
\end{equation*}
For the second term in~\eqref{eqn:FIM}, we have
\begin{equation*}
\frac{\partial \bm{\Sigma}_{\mathrm{e,d}}}{\partial \eta_i}
= -\frac{c_n}{4n} \cdot \frac{\partial \bm{v} \bm{v}^T}{\partial \eta_i}
= -\frac{c_n}{4n} \left( \frac{\partial \bm{v}}{\partial \eta_i} \cdot \bm{v}^T + \bm{v} \cdot \frac{\partial \bm{v}^T}{\partial \eta_i} \right)
= -\frac{c_n}{8n\sqrt{\eta_i}}(\bm{e}_i \bm{v}^T + \bm{v} \bm{e}_i^T).
\end{equation*}
For convenience, let $b = \frac{4c_n}{n-c_n S_{\bm{\eta}}}$.
Since $\bm{v}^T \bm{e}_i = \sqrt{\eta_i}$, and $\bm{v}^T \bm{v} = S_{\bm{\eta}}$,
\begin{align*}
\bm{\Sigma}_{\mathrm{e,d}}^{-1} \cdot \frac{\partial \bm{\Sigma}_{\mathrm{e,d}}}{\partial \eta_i}
& = -\frac{c_n}{8n\sqrt{\eta_i}} (4 \bm{I} + b\bm{v} \bm{v}^T)(\bm{e}_i \bm{v}^T + \bm{v} \bm{e}_i^T) \\
& = -\frac{c_n}{2n\sqrt{\eta_i}} \bm{e}_i \bm{v}^T -\frac{c_n(4+b S_{\bm{\eta}})}{8n\sqrt{\eta_i}} \bm{v} \bm{e}_i^T -\frac{bc_n}{8n} \bm{v} \bm{v}^T,
\end{align*}
\begin{align*}
\bm{\Sigma}_{\mathrm{e,d}}^{-1} \cdot \frac{\partial \bm{\Sigma}_{\mathrm{e,d}}}{\partial \eta_i} \cdot \bm{\Sigma}_{\mathrm{e,d}}^{-1} \cdot \frac{\partial \bm{\Sigma}_{\mathrm{e,d}}}{\partial \eta_j}
& = \left( 
\underbrace{\frac{c_n}{2n\sqrt{\eta_i}} \bm{e}_i \bm{v}^T}_{(\blacksquare)}
+ \underbrace{\frac{c_n(4+b S_{\bm{\eta}})}{8n\sqrt{\eta_i}} \bm{v} \bm{e}_i^T}_{(\blacktriangle)}
+ \underbrace{\frac{bc_n}{8n} \bm{v} \bm{v}^T}_{(\blacklozenge)} \right) \\
& \qquad \qquad \cdot 
\left( \underbrace{\frac{c_n}{2n\sqrt{\eta_j}} \bm{e}_j \bm{v}^T}_{(\square)} 
+ \underbrace{\frac{c_n(4+b S_{\bm{\eta}})}{8n\sqrt{\eta_j}} \bm{v} \bm{e}_j^T}_{(\vartriangle)} 
+ \underbrace{\frac{bc_n}{8n} \bm{v} \bm{v}^T}_{(\lozenge )} \right).
\end{align*}
Let $\mathbbm{1}[i=j]$ denote the indicator function, such that $\mathbbm{1}[i=j] = 1$ if and only if the condition $i=j$ holds, and $\mathbbm{1}[i=j] = 0$ otherwise. 
By the cyclic property of the trace, we obtain the following product terms,
\begin{align*}
& \Tr((\blacksquare)(\square))
= \frac{c_n^2}{4n^2\sqrt{\eta_i \eta_j}} \cdot \Tr(\bm{e}_i \bm{v}^T \bm{e}_j \bm{v}^T)
= \frac{c_n^2}{4n^2\sqrt{\eta_i \eta_j}} \cdot \Tr(\bm{v}^T \bm{e}_j \bm{v}^T \bm{e}_i)
= \frac{c_n^2}{4n^2}, \\
& \Tr((\blacksquare)(\vartriangle))
= \frac{c_n^2(4+b S_{\bm{\eta}})}{16n^2\sqrt{\eta_i\eta_j}} \cdot \Tr(\bm{e}_i \bm{v}^T \bm{v} \bm{e}_j^T) = \frac{c_n^2S_{\bm{\eta}}(4+b S_{\bm{\eta}})}{16n^2\sqrt{\eta_i\eta_j}} \cdot \Tr(\bm{e}_i \bm{e}_j^T)
= \frac{c_n^2S_{\bm{\eta}}(4+b S_{\bm{\eta}})}{16n^2\sqrt{\eta_i\eta_j}} \cdot \mathbbm{1}[i=j], \\
& \Tr((\blacksquare)(\lozenge))
= \frac{bc_n^2}{16n^2\sqrt{\eta_i}} \cdot \Tr(\bm{e}_i \bm{v}^T \bm{v} \bm{v}^T)
= \frac{bc_n^2 S_{\bm{\eta}}}{16n^2\sqrt{\eta_i}} \cdot \Tr(\bm{e}_i \bm{v}^T)
= \frac{bc_n^2 S_{\bm{\eta}}}{16n^2\sqrt{\eta_i}} \cdot \bm{v}^T \bm{e}_i
= \frac{bc_n^2 S_{\bm{\eta}}}{16n^2},\\
& \Tr((\blacktriangle)(\square))
= \frac{c_n^2(4+b S_{\bm{\eta}})}{16n^2\sqrt{\eta_i\eta_j}} \cdot \Tr(\bm{v} \bm{e}_i^T \bm{e}_j \bm{v}^T) 
= \frac{c_n^2(4+b S_{\bm{\eta}})}{16n^2\sqrt{\eta_i\eta_j}} \cdot \mathbbm{1}[i=j] \cdot \bm{v}^T \bm{v}
= \frac{c_n^2 S_{\bm{\eta}} (4+b S_{\bm{\eta}})}{16n^2\sqrt{\eta_i\eta_j}} \cdot \mathbbm{1}[i=j], \\
& \Tr((\blacktriangle)(\vartriangle))
= \frac{c_n^2(4+b S_{\bm{\eta}})^2}{64n^2\sqrt{\eta_i \eta_j}} \cdot \Tr(\bm{v} \bm{e}_i^T \bm{v} \bm{e}_j^T)
= \frac{c_n^2(4+b S_{\bm{\eta}})^2}{64n^2\sqrt{\eta_j}} \cdot \Tr(\bm{e}_j^T \bm{v}) 
= \frac{c_n^2(4+b S_{\bm{\eta}})^2}{64n^2}, \\
& \Tr((\blacktriangle)(\lozenge))
= \frac{bc_n^2(4+b S_{\bm{\eta}})}{64n^2\sqrt{\eta_i}} \cdot \Tr(\bm{v} \bm{e}_i^T \bm{v} \bm{v}^T)
= \frac{bc_n^2(4+b S_{\bm{\eta}})}{64n^2} \cdot \Tr(\bm{v} \bm{v}^T)
= \frac{bc_n^2 S_{\bm{\eta}} (4+b S_{\bm{\eta}})}{64n^2}, \\
& \Tr((\blacklozenge)(\square))
= \frac{bc_n^2}{16n^2\sqrt{\eta_j}} \cdot \Tr(\bm{v} \bm{v}^T \bm{e}_j \bm{v}^T)
= \frac{bc_n^2}{16n^2} \cdot \Tr(\bm{v} \bm{v}^T)
= \frac{bc_n^2 S_{\bm{\eta}}}{16n^2}, \\
& \Tr((\blacklozenge)(\vartriangle))
= \frac{bc_n^2(4+b S_{\bm{\eta}})}{64n^2\sqrt{\eta_j}} \cdot \Tr(\bm{v} \bm{v}^T \bm{v} \bm{e}_j^T)
= \frac{bc_n^2 S_{\bm{\eta}} (4+b S_{\bm{\eta}})}{64n^2},\\
& \Tr((\blacklozenge)(\lozenge))
= \frac{b^2 c_n^2}{64n^2} \cdot \Tr(\bm{v} \bm{v}^T \bm{v} \bm{v}^T)
= \frac{b^2 c_n^2 S_{\bm{\eta}}^2}{64n^2}.
\end{align*}
Using the linearity of the trace, we get the second term in~\eqref{eqn:FIM} by summing up all of the   above,
\begin{equation} \label{eqn:I_{e,d}_{i,j}_term2}
\frac{1}{2} \Tr \left( \bm{\Sigma}_{\mathrm{e,d}}^{-1} \cdot \frac{\partial \bm{\Sigma}_{\mathrm{e,d}}}{\partial \eta_i} \cdot \bm{\Sigma}_{\mathrm{e,d}}^{-1} \cdot \frac{\partial \bm{\Sigma}_{\mathrm{e,d}}}{\partial \eta_j} \right)
= \frac{c_n^2(n+c_n S_{\bm{\eta}})}{4n(n-c_n S_{\bm{\eta}})^2}
+ \frac{c_n^2 S_{\bm{\eta}}}{4n(n-c_n S_{\bm{\eta}})} \cdot \frac{\mathbbm{1}[i=j]}{\sqrt{\eta_i \eta_j}}.
\end{equation}
The claim follows by combining \eqref{eqn:I_{e,d}_{i,j}_term1} and \eqref{eqn:I_{e,d}_{i,j}_term2}, which gives
\begin{equation*}
(\mathcal{I}_{\mathrm{e,d}})_{i,j}
= \frac{c_nN}{n-c_nS_{\bm{\eta}}} + \frac{c_n^2(n+c_n S_{\bm{\eta}})}{4n(n-c_n S_{\bm{\eta}})^2} 
+ \left( N + \frac{c_n^2 S_{\bm{\eta}}}{4n(n-c_n S_{\bm{\eta}})} \right) \cdot \frac{\mathbbm{1}[i=j]}{\sqrt{\eta_i \eta_j}}.
\end{equation*}